\documentclass[12pt]{article}
\usepackage[OT1]{fontenc}
\usepackage[onehalfspacing]{setspace}
\RequirePackage{amsthm}
\RequirePackage[cmex10]{amsmath}
\usepackage[cmyk,x11names]{xcolor}

\RequirePackage[colorlinks=false,hidelinks]{hyperref} 

\usepackage[top=1.25 in, bottom=1.25 in, left=1.25in, right=1.25in]{geometry}
\usepackage{epsfig,amsfonts,setspace,dsfont,array,wasysym,enumerate,sgame,subcaption,comment,amssymb,tikz,slantsc,ifthen,xfrac,natbib,epstopdf,eqnarray,mathtools,xspace,nicefrac}

\usepackage[shortlabels]{enumitem}

\usepackage[normalem]{ulem}
\usepackage[cmtip,all]{xy}
\usepackage{mathtools}
\usepackage[charter]{mathdesign}

\usepackage[capitalise,noabbrev]{cleveref}
\crefname{section}{section}{sections}
\newtheorem{innercustomgeneric}{\customgenericname}
\providecommand{\customgenericname}{}
\newcommand{\newcustomtheorem}[2]{%
	\newenvironment{#1}[1]
	{%
		\renewcommand\customgenericname{#2}%
		\renewcommand\theinnercustomgeneric{##1}%
		\innercustomgeneric
	}
	{\endinnercustomgeneric}
}

\newtheorem{lemma}{Lemma}
\newtheorem{claim}{Claim}
\newtheorem*{claim*}{Claim}
\newtheorem{theorem}{Theorem}

\newtheorem{corollary}{Corollary}

\newtheorem{proposition}{Proposition}
\newtheorem{definition}{Definition}

\newcustomtheorem{customthm}{Theorem}

\theoremstyle{definition}

\newtheorem*{example*}{Example}
\newtheorem{mainexample}{Example}
\newcommand{\examplepart}[1]{\noindent\textbf{Example~\ref{pd}#1.} }

\newcommand{\supp}{\textnormal{supp }}
\newcommand{\eps}{\varepsilon}

\newcommand{\rmd}{\mathrm{d}}
\newcommand{\bbE}{\mathbb{E}}

\usepackage[compact,small,raggedright]{titlesec}
\titlelabel{\thetitle.\quad}

\begin{document}
	
	\begin{titlepage}
		\title{{\bf \Large{Coarse Information Design}}\thanks{
				We are especially grateful to Gregorio Curello for his continuous support and insightful discussions. We also thank Sarah Auster, Yunus Aybas, Dirk Bergemann, Florian Brandl, Francesc Dilm\'{e} , Michael Greinecker, Li Hao, \'{A}ngel Hernando-Veciana, Ilwoo Hwang, Mathijs Janssen, Andreas Kleiner, Jin Li, Benny Moldovanu, Paula Onuchic, Ali Shourideh, Aloysius Siow, Peter S{\o}rensen, Junze Sun, Dezs\H{o} Szalay, Lina Uhe, Kun Zhang, Chen Zhao, and participants at Bonn, CUHK, NTU, NUS, UC3M, Vienna, Washington, the 2023 Stony Brook International Conference on Game Theory, EWMES 2023, the 2024 Conference on Mechanism and Institution Design and VIEE for valuable suggestions and comments on this topic. Lyu acknowledges support from the Deutsche Forschungsgemeinschaft (DFG, German Research Foundation) under Germany's Excellence Strategy EXC 2126/1-390838866. Zhang acknowledges support from the Deutsche Forschungsgemeinschaft (DFG, German Research Foundation) under Germany's Excellence Strategy EXC 2047/1-390685813.	
			}
		}
		\vspace{2cm}
		\author{
			\begin{minipage}{0.3\textwidth}\centering  
				Qianjun Lyu
				\\ \centering  \small \it University of Bonn
			\end{minipage}  
			\begin{minipage}{0.3\textwidth}\centering 
				Wing Suen
				\\ \centering  \small \it University of Hong Kong
			\end{minipage}  
			\begin{minipage}{0.3\textwidth}\centering 
				Yimeng Zhang
				\\ \centering  \small \it University of	Bonn
			\end{minipage}  
		}
		
		\date{\vspace{0.5cm} \today \\
			\vspace{0.8cm}}
		% \href{https://arxiv.org/abs/2305.18020}{\textcolor{blue}{Latest version is available here}}.}
	\maketitle
	\thispagestyle{empty}
	\vspace{-0.8cm}
	
	\begin{quote}

		\noindent \textit{Abstract:}
		We study an information design problem with continuous state and discrete signal space. Under convex and S-shaped value functions, the optimal information structure is interval-partitional and exhibits a dual expectations property: each induced signal is the conditional mean (taken under the prior density) of each interval; and each interval cutoff is the barycenter (taken under the value function curvature) of the interval formed by neighboring signals.  This property enables an examination into which part of the state space is more finely partitioned. The analysis can be extended to general value functions and	adapted to study coarse mechanism design.
		
		\noindent 
		
		\vspace{0.5cm}
		
		\vspace{0.5cm}
		
		
	\end{quote}

\end{titlepage}

\newpage 

\section{Introduction} 
\label{section:Introduction}

Coarse information is ubiquitous. Teachers evaluate students by letter grades or on a pass/fail basis; online platforms often convey product reviews through coarse rating systems; and regulatory policies disclose efficiency and safety level of different products via discrete categories. 
These coarse information mechanisms can emerge due to limited cognitive and memory capacities, imperfect communication channels, technological constraints on the measurement instruments, or simply for the sake of convenience.  
In pain medicine, the intensity of pain naturally falls into a continuum, but it would be almost impossible for patients to fully describe their subjective feeling or convert it into a real number on a continuous scale.  Instead, doctors rely on instruments such as the Numerical Rating Scale, from 0 to 10, based on patients' self-reported pain level to reach a diagnosis.  In an organizational context, individuals at different levels of the hierarchy may not have the time or the expertise to digest complicated information from each other if the communication involves exhaustive details.  In practice, they typically adopt language protocols that are less precise but more comprehensible to facilitate decision-making.  
Furthermore, because different individuals in the organization are making different decisions, the optimal way to coarsen information is tailored to different needs. An executive summary of a research report written for the head of engineering, for example, should be quite different from a summary of the same report written for the CEO.
In this paper, we take the coarse nature of information structure as given and study how to design information optimally subject to such a constraint.

Formally, we study a class of information design problems where the uncertain state $\theta$ is continuously distributed on $[0,1]$ and is payoff relevant only through its expectation.
The (information) designer---or the sender---commits to a Blackwell experiment $(\pi,\Sigma)$, where $\Sigma$ is the signal space and $\pi(\sigma|\theta)$ is the probability of obtaining signal realization $\sigma \in \Sigma$ conditional on state $\theta$. The receiver updates her belief based on the realized signal and then chooses an action optimally. 
Our model allows sender and receiver to the be same person. In that case, the relevant value function is naturally convex, and an unconstrained information design problem would be trivial as the optimal information structure is fully informative.
The constraint we introduce in this paper is that $\Sigma$ is restricted to be a finite set with a cardinality of at most $N$.  Given that the state is continuous and the signals are finite, 
a fully informative experiment is infeasible.  
In other words, a coarse experiment necessarily involves pooling across different states. 
How to allocate the limited ``signal resources'' then becomes a relevant economic question---which
parts of the state space should receive \emph{closer scrutiny} in the experiment relative to other parts? 

In this introduction, unless otherwise specified, 
we illustrate our findings under scenarios in which the sender and receiver's preferences are sufficiently aligned such that more information is strictly beneficial (i.e., the sender's value function $u$ as a function of the induced posterior mean is convex). The insights extend to a broader class of value functions.

Consider an \emph{interval-partitional} information structure, where an experiment divides the state space into a finite number of subintervals and reveals which subinterval the state belongs to.
We say that a subinterval of the state space receives ``closer scrutiny'' than others if the width of this subinterval is smaller than that of others. In Figure \ref{AvsB}, we represent an experiment by the set of cutoffs that defines the partition. Under experiment $A$, the state space in the neighborhood of 0.7 receives the closest scrutiny.  In contrast, experiment $B$ explores the state space near 0.3 with the
closest scrutiny.  
We refer to the subinterval that receives the closest scrutiny as the ``center of scrutiny.''
In both experiments, the interval becomes progressively wider farther away from the center of scrutiny.

\bigskip
\begin{figure}[h]
	\centering
	\begin{subfigure}[b]{0.49\textwidth}
		\centering
		\includegraphics[width=\textwidth]{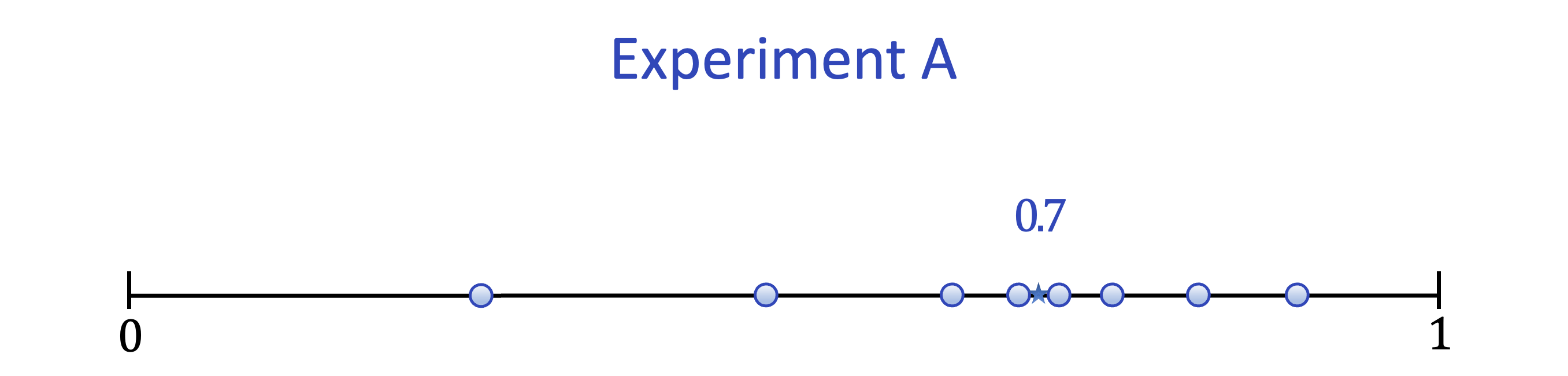}
		\label{}
	\end{subfigure}
	\hfil
	\begin{subfigure}[b]{0.49\textwidth}
		\centering
		\includegraphics[width=\textwidth]{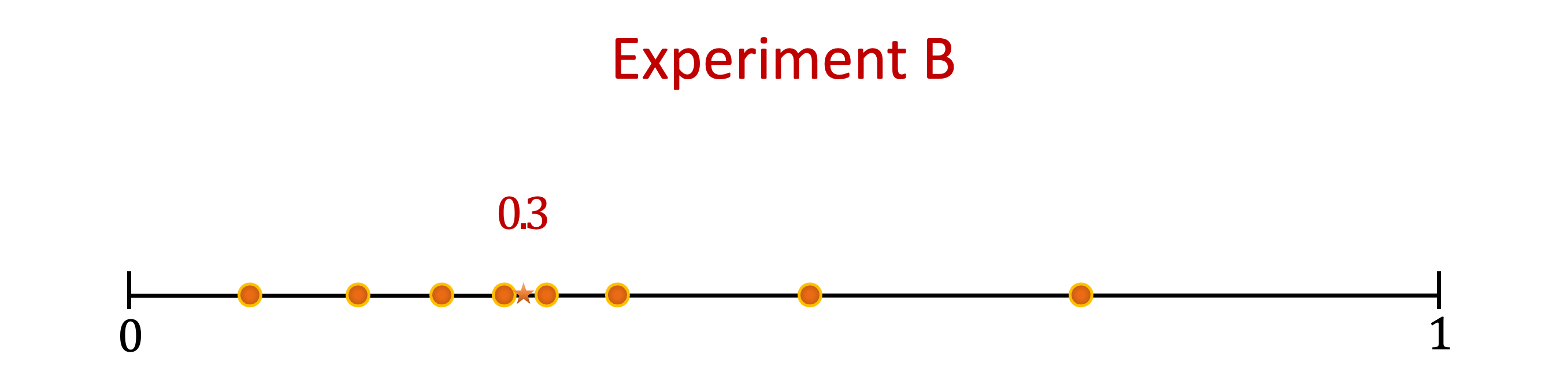}
		\label{}
	\end{subfigure}
	\caption{\small  Experiment $A$ gives close scrutiny to the state space near 0.7; experiment $B$ gives close scrutiny to the state space near 0.3.}
	\label{AvsB}
\end{figure}

Naturally, the designer should give closer scrutiny to states that are more likely to happen (according to the prior distribution), because having the receiver make informed decisions under those states carries a greater weight in his ex ante expected utility. This insight is developed  in other applications such as quantization theory \citep{GrayNeuhoff1988} or preference evolution \citep{Robson2001,Netzer2009}.
On the other hand, the designer should also give closer scrutiny to those states where local information is more valuable (i.e., where the value function has greater curvature, measured by the second derivative $u''$ of the value function). 
The latter concern has received less attention in the existing literature, let alone the connection between these two forces.

Our first main result characterizes the designer's optimal interval-partitional information structure and establishes a novel dual relation between the prior density $f$ and the curvature of the value function $u''$. 
For each interval-partitional experiment, represented by the cutoffs
$\{s_k\}_{k=0}^{N}$, the signal $x_k$ induced within the interval $(s_{k-1},s_k)$ equals the expected value of this interval, computed using the conditional density of the prior distribution, 
$f(\cdot)/(F(s_k)-F(s_{k-1}))$.  
Under the \textit{optimal} 
interval partition, each cutoff point $s_k$ equals the ``expected value'' of the interval formed by two adjacent signals  $(x_k , x_{k+1})$,  computed using the conditional 
``density,'' 
$u''(\cdot)/\left(u'\left(x_{k+1}\right)-u'\left(x_k\right)\right)$.  
We call such a property of the optimal experiment \textit{dual expectations}.

The dual expectations property
enables us to characterize the optimal experiment as the fixed point of a monotone equation system.
The second key finding highlights an additional feature of the optimal interval partition: In a \textit{logconcave} environment---where both the prior density  $f$ and the curvature $u''$ are 
logconcave---the widths of the subintervals are \emph{single-dipped}. Specifically, there is a center of scrutiny---a subinterval that receives the closest scrutiny.  As we move away from this center, the level of scrutiny diminishes for states located farther away (as illustrated in  Figure \ref{AvsB}). 
This
pattern of information structure is commonly observed in various mandatory energy-efficiency rating standards, such as Energy Performance Certificates for residential buildings in Austria, France, Germany, and Italy.\footnote{
See \url{https://www.stanz-landeck.gv.at/system/web/getDocument.ashx?ncd=1&ltc=1&fileid=1294684} (page 5), \url{https://www.ecologie.gouv.fr/sites/default/files/documents/Guide_pour_les_diagnostiqueurs_DPE.pdf} (page 29), \url{https://www.gesetze-im-internet.de/geg/anlage_10.html}, and \url{https://www.mimit.gov.it/images/stories/normativa/DM_Linee_guida_APE_allegato1.pdf} (Table 2).}

We leverage the dual expectations property to facilitate comparative statics analysis. When either the prior density $f$ or the curvature function $u''$ adopts a likelihood-ratio increase, 
all the interval cutoff points 
$\{s_k\}_{k=1}^{N-1}$ and the posterior means 
$\{x_k\}_{k=1}^{N}$ will shift to the right.  
Furthermore, if either $f$ or $u''$ becomes less variable according to the \emph{uniform conditional variability order} \citep{Whitt1985}, then both the interval cutoff points and the posterior means become more compressed, in the sense that there exists an $n^*$ such that $s_k$ shift to the right for all $k < n^*$ and $s_k$ shifts to the left for all $k> n^*$. 

With convex value function, our framework has a natural application to mechanism design where the agent's payoff is linear in 
her private information and the contract is constrained to be finite. In Section \ref{nonlinearpricing}, we illustrate how to transform a standard nonlinear pricing problem with a finite menu into a coarse information design problem. Most of our results are transferable to study the design and properties of the optimal finite menu. 

In Sections \ref{section:s-shaped} and \ref{generalsection}, we extend our analysis to S-shaped value functions and more general value functions to address scenarios where information could potentially hurt the sender
(i.e., $u''<0$ for some regions of the state space). With S-shaped value functions, the dual expectations characterization, the single-dipped property, as well as the comparative statics results, remain
true. In particular, the dual expectations property is modified with respect to a signed measure---the optimal cutoff point $s_k$ equals the \emph{barycenter} of the interval formed by two adjacent signals $(x_k,x_{k+1})$ under the signed measure $u'$. 
With more general value functions, we restrict our attention to the optimal interval-partitional information structure, and 
show that the dual expectations property remains valid, incorporating the potential usage of a bi-tangency condition when the barycenter is not well defined. 

We provide further discussions of our model in Section \ref{Discussion}.  Section \ref{section:piece-wiseaffine} discusses the scenario where there are only finite actions and the number of messages is smaller than the number of actions. 
Section \ref{section:Beyond} uses numerical examples beyond logconcave environments to demonstrate the robustness of our conclusions regarding the center of scrutiny. 
In Section \ref{section:Cheaptalk} we establish connections between our model and cheap talk. 

\textbf{Related literature.} 
Our paper contributes to the literature on information design with a continuous state space. We adopt the approach introduced by \cite{GentzkowKamenica2016b}, which represents an information structure by the integral of its cumulative distribution function.  \cite{Kolotilin2018} and \cite{DworczakMartini2019} use alternative approaches to studying this problem. \cite{KleinerMoldovanuStrack2021} and \cite{ArieliEtAl2023} show that the optimal unconstrained information structure exhibits a
bi-pooling property. 
\cite{CurelloSinander2023} study the comparative statics of 
mean-measurable persuasion problems and identify the conditions under which a sender with ``more convex'' value function will design a more informative signal structure. 
The bulk of the Bayesian persuasion literature starting from \cite{KamenicaGentzkow2011} and \cite{RayoSegal2010} is concerned about the strategic use of pooling: how the sender designs an experiment to ``concavify'' his value function by strategically pooling information to influence the receiver's action.  When the signal space is constrained to be finite, pooling becomes necessary even when the value function is convex.  Our paper examines how to effectively pool information in such environments. 

The constraint imposed by coarse information on the operation of markets is discussed by \cite{Wilson1989}, \cite{McAfee2002}, \cite{HoppeMoldovanuOzdenoren2010} and \cite{Kos2012}.  \cite{Dow1991} considers a sequential search problem for a decision maker whose memory is represented by a partition of the set of possible past prices. \cite{HarbaughRasmusen2018} and \cite{OstrovskySchwarz2010} analyze coarse grading from an information design perspective,	
while \cite{CremerGaricanoPrat2007} study coarse language code under a discrete environment where the loss function only depends on the number of states that are pooled together in one message.

Recently, the literature on information design has begun to explore the limitation of communication and information channels.
\cite{OnuchicRay2022} and \cite{Mensch2021} study information design subject to monotone information structures. 
\cite{AybasTurkel2021} examine an information design problem under a bounded signal space. They characterize the highest achievable sender payoff and analyze how it changes with the cardinality of the signal space. Their paper explores finite state and finite action space, and allows the sender's payoff to depend arbitrarily on beliefs. In contrast, we focus specifically on mean-measurable persuasion problems and provide a general characterization for the optimal information structure.

Our paper is related to \cite{HopenhaynSaeedi2024}, 
which studies the optimal coarse ratings system that partitions a continuum of sellers with different qualities into a finite number of quality groups in a competitive environment.  They characterize the social planner's optimal rating scheme and 
compare the rating schemes under different shapes of supply functions.
\cite{Tian2022} uses ``cell functions'' (a mapping from  subintervals to payoffs) as primitives to study the optimal interval partition problem. 
He focuses on submodular cell functions, and shows that the optimal interval cutoffs shift up when the prior distribution shifts up. 
We use interim value functions as primitives and characterize the optimal information structure for general value functions. 
Moreover, the dual relation between value function curvature and prior density allows us to develop comparative statics with respect to changes
in payoff functions.

The coarse information design problem addressed in this paper resembles quantization in information theory \citep{GrayNeuhoff1988, MeaseNair2006}. These two classes of problems are not nested. The main difference is that quantization algorithms typically maximize the expected ``similarity'' within a group based on a distance metric, e.g., minimizing the expected distance between the realized values and the centroid within a cluster, which in general, cannot be transformed into a well-defined interim value function with a domain of posterior means. One exception is when the loss function is quadratic, in which case the quantization problem becomes equivalent to our problem with uniform value function curvature.

\section{Model}
\label{section:Model}

Consider an information design model where the payoff-relevant state  $\theta$ is drawn from a prior distribution $F$ on $[0,1]$. 
We assume that $F$ admits a continuous density function $f$ with full support. A designer commits to an information structure: a signal space $\Sigma$ and a mapping $\pi: [0,1] \to \Delta (\Sigma)$ from the state to a distribution over signals, which induces a random posterior. We focus on models where the designer's interim payoff depends on the induced posterior belief only through the posterior mean,\footnote{
	As \cite{DworczakMartini2019} point out, the posterior-mean dependent setup is satisfied when the receiver's optimal action only depends on the expected state and the sender's preference over action is linear on the state. Specifically, the receiver's utility function can be expressed as $y_1(a)+y_2(a) \theta+y_3(\theta)$ and the designer's utility function can be represented by $z_1(a)+z_2(a) \theta$, where $a$ is an element from a compact set $A$ and $y_1,y_2,y_3,z_1,z_2$ are upper semi-continuous functions. \cite{Ivanov2021} observes that if the designer's utility function includes a third term dependent solely on the state $z_3(\theta)$, the optimal experiment remains unchanged.
}   
and denote the interim value function by $u: [0,1] \to \mathbb{R}$.
Throughout the paper, we maintain the assumption that $u$ is twice continuously differentiable and 
$u''$ is not equal to zero almost everywhere.
In Section \ref{section:piece-wiseaffine}, we consider value functions that are piecewise affine, a situation that often arises when the underlying action space is discrete. 
Since the only relevant information is the posterior mean, it is without loss of generality to assume that the realized signal is the posterior mean itself.  From now on, we will primarily work with the induced distribution of posterior means, denoted by $G$. 

We follow \cite{GentzkowKamenica2016b} by representing an information structure as the integral of the induced distribution function of posterior means. Let $F_0$ be the degenerate distribution that puts probability mass one on the prior mean of $F$. For any distribution $G$ on $[0,1]$, define $I_G(x):=\int_0^x G(t)\,\rmd t$ and call it the \emph{integral distribution} of $G$.
By \citeauthor{Strassen1965}'s (\citeyear{Strassen1965}) theorem, $G$ can be induced by some signal structure if and only if $I_{F_0}\le I_G\le I_F$; namely, $G$ must be a mean-preserving contraction of $F$.

The main ingredient of our model is the introduction of a capacity constraint on the signal space $\Sigma$. In particular, we require $\Sigma$ to contain no more than $N$ elements. Consequently, the distribution $G$ of posterior means can only have a finite support.
We can then write the designer's optimization program as:
\begin{align}
	\label{primalprogram}
	&\max_{G \in \triangle([0,1])} \, \int_{0}^1 u(x) \, \rmd G(x)\\
	&\quad\text{s.t.} \qquad I_{F_0}\le I_G\le I_F,  \tag{MPC} \label{MPCconstraint}\\
	&\quad\phantom{\text{s.t.}} \qquad \  |\supp (G)|\le N. \tag{D} \label{Dconstraint}
\end{align} 

A distribution with finite support has an integral distribution which is increasing, convex and piecewise-affine, with kinks at every element in the support of the distribution. Let 
$\text{ICPL}$ denote the set of such integral distribution functions defined on $[0,1]$ that satisfy the mean-preserving contraction constraint (\ref{MPCconstraint}). 
For every $I \in \text{ICPL}$, define the set of kink points of $I$ by $\mathcal{K}_I := \{ x\in (0,1) : I'(x^-) \neq I'(x^+) \}$.

We now transform the objective function (\ref{primalprogram}) 
into an explicit integral of an ICPL function:
\begin{align*}
	\int_0^1 u(x)\,\rmd G(x) 
	&= u(1)-u'(1)I_F(1) + \int_0^1 u''(x)I_G(x) \, \rmd x
\end{align*} 
The first two terms are constants. Hence, the original problem can be rewritten as:
\begin{align}
	&\max_{I_G \in \text{ICPL}}\, \int_{0}^1 u''(x)I_G(x) \, \rmd x \label{auxiliaryprogram}\\
	&\quad\text{s.t.} \qquad  \left|\mathcal{K}_{I_G}\right|\le N. \tag{D} 
\end{align} 
The objective function (\ref{auxiliaryprogram}) represents the ``signed weighted'' area under the integral distribution function, where the signed weights are given by the \emph{curvature} of the value function, $u''$.  Intuitively, the level of an integral distribution $I_G$ measures how informative the corresponding information structure is locally. A higher $I_G$ is closer to the full-information integral distribution
$I_F$. 
It yields a higher payoff when providing local information is very rewarding---i.e., when associated with higher positive $u''$.  Similarly, a lower $I_G$ is closer to the 
null-information integral distribution
$I_{F_0}$, and is more preferable to the designer if $u''$ is negative.
Obviously, if $u$ is convex within some region, a locally upward shift of $I_G$ increases the designer's expected payoff, provided that such a shift does not exceed $I_F$. Moreover, a clockwise rotation of a certain segment of $I_G$ at the point when $u$ switches from convex to concave is beneficial, provided that $I_G$ remains in the ICPL class.

Throughout the paper, we repeatedly visit the following class of applications to illustrate the results. 

\begin{mainexample}[Purchase Decision]\label{pd}
	The state $\theta \in [0,1]$ is the value of a good, and the price of the good is $p$.  The decision maker chooses to buy the good ($a=1$) or not ($a=0$).  She first observes a signal from some experiment to learn about $\theta$. After the signal is
observed but before the purchase decision is made, a cost shock $\eta$ is realized, where $\eta$ is distributed with density $h$ on $[\underline{\eta},\overline\eta]$ with $\underline{\eta}\le -p$, which can be interpreted as either the actual cost of consuming the good or the opportunity cost of outside option. When her posterior expectation of the state is $x$, she buys the good if and only if $x-p-\eta \ge 0$.  Her interim utility function is her expected consumer surplus, $u(x)=\int^{x-p}_{\underline\eta} (x-p-\eta)h(\eta)\,\rmd \eta$, which is a convex function with curvature $u''(x)=h(x-p)$. On the other hand, the seller wants to maximize the probability of sale. Hence, his interim payoff is the expected revenue $\hat u(x)=p H(x-p)$ with curvature $\hat u''(x)=p h'(x-p)$. 	The example encapsulates two sets of models: when the \textit{buyer} designs the information structure, the framework corresponds to the model where a single agent designs a coarse experiment to make a better-informed decision in the future, pioneered by \cite{Dow1991}; when the \textit{seller} designs the information structure, it corresponds to 
a problem of persuading a privately informed receiver \citep{KolotilinEtAl2017}. 	The former case features convex value functions and is discussed in \Cref{section:Convex}, while the latter case features S-shaped value functions when the cost shock follows a 
unimodal density and is discussed in \Cref{section:s-shaped}.
\end{mainexample}

A natural class of coarse information structures is
the set of interval-partitional information structures. That is, we partition the state space $[0,1]$ into $N$ subintervals, each of which sends a distinct signal. Interval partitions are widely observed in contexts such as certification or grading schemes, where signals naturally correspond to intervals over an underlying quality or performance scale.  They are
easier to implement than general signal structures that may require randomization or involve pooling states from disparate parts of the state space. Motivated by these considerations, in the main text of this paper, we restrict our attention 
to this particular class of information structures. In most of this paper, the solution to this problem also solves the more general problem without such a restriction.
We will be clear when this issue arises. 

For an interval-partitional information structure $ G$ defined by cutoff points, $ 0 = s_0 < s_1 < \ldots < s_N = 1 $, its integral representation is illustrated in Figure~\ref{interval}. The tangency points between $ I_F $ and $ I_G $ correspond to these cutoff points because $G(s_k)=I'_G (s_{k})=I'_F(s_{k})=F(s_k)$ for all $k=0,\ldots,N$. Moreover
each kink point $x_k$ in $ I_G $ between two adjacent tangency points $s_{k-1}$ and $s_k$
represents the posterior mean of the state over the corresponding subinterval $(s_{k-1}, s_k)$.

\begin{lemma}
	\label{lemma:conditional_expectation}
	For any $I_G \in \mathrm{ICPL}$, if $x_k$ is the only kink point between two adjacent tangency points $s_{k-1}$ and $s_{k}$, then ${x_k}=  \mathbb{E}_F\left[ \theta\, \middle \vert\, \theta \in (s_{k-1},s_k)\right]$.
\end{lemma}

\begin{proof}
	The kink point $x_k$ satisfies the equation,
	\[
	I_F(s_{k-1}) + F(s_{k-1})(x_k- s_{k-1})= I_F(s_k)+F(s_k)(x_k - s_k).
	\]
	Solving this equation gives:
	\[
	x_k = \frac{\int_{s_{k-1}}^{s_k} \theta \, \rmd F(\theta)}{\int_{s_{k-1}}^{s_k} \rmd F(\theta)}.\qedhere
	\]
\end{proof}

\begin{figure}[t]
	\centering
	\includegraphics[width=7cm]{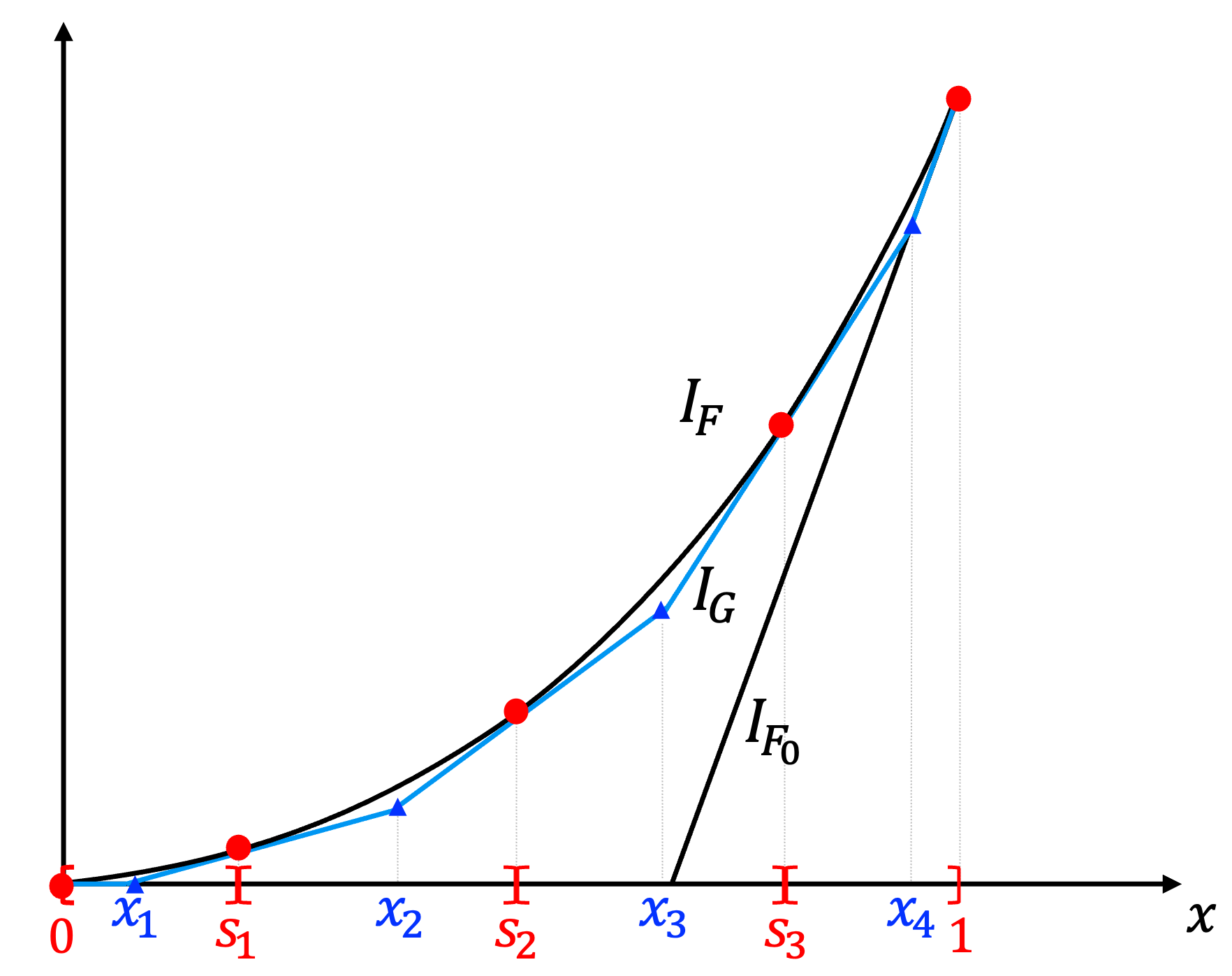}
	\caption{\small The integral distribution function of an interval partition with 4 signals.	}
	\label{interval}
\end{figure}

In fact, for any $N$-interval partition $G$, it is \textit{essentially} the only information structure to induce the corresponding distribution of posterior means represented by the kink points of $I_G$. Formally, if we consider all the signals that induce the same posterior mean an 
equivalence class, then the $N$-interval partition $G$ is the unique information structure that implements $I_G$.

\section{Convex Value Functions}
\label{section:Convex}

In this section, we discuss convex value functions, where sender and receiver's preferences are sufficiently aligned such that more information is strictly beneficial. The finiteness constraint raises an important question regarding how to pool information efficiently when the number of disposable ``signal resources'' is budgeted. Moreover, this class of value functions incorporates many natural applications, such as a single-agent decision problem or a sender-receiver game with quadratic utilities. 

Upon normalization, it is without loss of generality in this case to treat $u''$ as a density function that captures the relative value of information provision. Our main characterization of the optimal interval partition features a dual relation between the prior density $f$ and the curvature density $u''$. 
For any $0 \le a < b \le 1$, let
\begin{align*}
	\phi(a,b) & := \bbE_{F}\left[ t\, \middle \vert\, t \in \left( a,b \right)\right],\\
	\mu(a,b) & := \bbE_{u'}\left[ t\, \middle \vert \, t \in \left(a,b \right)\right] . 
\end{align*}
Here, $\bbE_{u'}$ is the conditional expectation operator using $u''(\cdot)/(u'(b)-u'(a))$ as the conditional density function.  When $u$ is strictly convex, this is a valid density of full support, and the ``conditional expectation'' function $\mu(\cdot)$ is well defined. 

\begin{theorem}
	\label{dualexpectation}
	Suppose $u$ is strictly convex. The optimal   $N$-interval partition $G$, characterized by $\{ s_k \}_{k=0}^{N}$ and $\{ x_k \}_{k=1}^{N}$, must satisfy:
	\begin{align}
		x_k &= \phi(s_{k-1},s_{k}) \qquad  \text{for } k=1,\ldots,N; \tag{CE-$F$} \label{phi}\\
		s_k &= \mu(x_k,x_{k+1}) \qquad \text{for } k=1,\ldots,N-1.
		\tag{CE-$u'$}\label{mu}
	\end{align}
\end{theorem}

\begin{proof}
	First, (\ref{phi}) must hold by definition of an interval-partitional information structure. To prove (\ref{mu}), for any interval-partitional structure $G$, we define the \textit{minorant function} $\underline{u}_G$ on each (open) segment as follows:
	\begin{equation}
		\underline{u}_G(x) := \sum_{k=1}^N \left[ u(x_k)+u'(x_k)(x-x_k) \right] \mathbb{I}_{\left(s_{k-1},s_k  \right)} (x).
		\label{minorantu}
	\end{equation}
On each subinterval $(s_{k-1},s_k)$, $\underline{u}_G(\cdot)$ is the subgradient of $u(\cdot)$ at the posterior mean of that subinterval. By construction, a minorant function is piecewise affine with increasing slopes. 
	
	Under an optimal interval partition $G^*$, the minorant function must be \emph{continuous} at the interval cutoff points.
Suppose to the contrary that $\underline{u}_{G^\ast} (s_k^-) < \underline{u}_{G^\ast} (s_k^+)$	(see left panel of Figure \ref{figure:proofdualexp}).	Then we can construct another function,
	\begin{equation*}
		\hat{\underline u}(x)= \begin{cases}
			\max\left\{u(x_k)+u'(x_k) (x-x_k),\, u(x_{k+1})+u'(x_{k+1}) (x-x_{k+1})\right\} &\text{if } x\in (s_{k-1},s_{k+1}),\\
			\underline{u}_{G^\ast}(x) &\text{otherwise}.
		\end{cases}
	\end{equation*}
	By construction $\hat{\underline u}(x)$ is continuous at $x=\hat s_k$, where $\hat{s}_k$ satisfies
	$u(x_k)+u'(x_k) (\hat s_k-x_k)=u(x_{k+1})+u'(x_{k+1}) (\hat s_k-x_{k+1})$. It is immediate that $\hat{s}_k \in (s_{k-1},s_k)$
(see right panel of Figure \ref{figure:proofdualexp}). Moreover, $\hat{\underline u}(x) \ge \underline u_{G^*}(x)$ for all $x$ in each open segment, and the inequality holds strictly for $x \in (\hat s_k, s_k)$.

	\begin{figure}[t]
		\centering
		\begin{subfigure}[b]{0.4\textwidth}
			\centering
			\includegraphics[width=\textwidth]{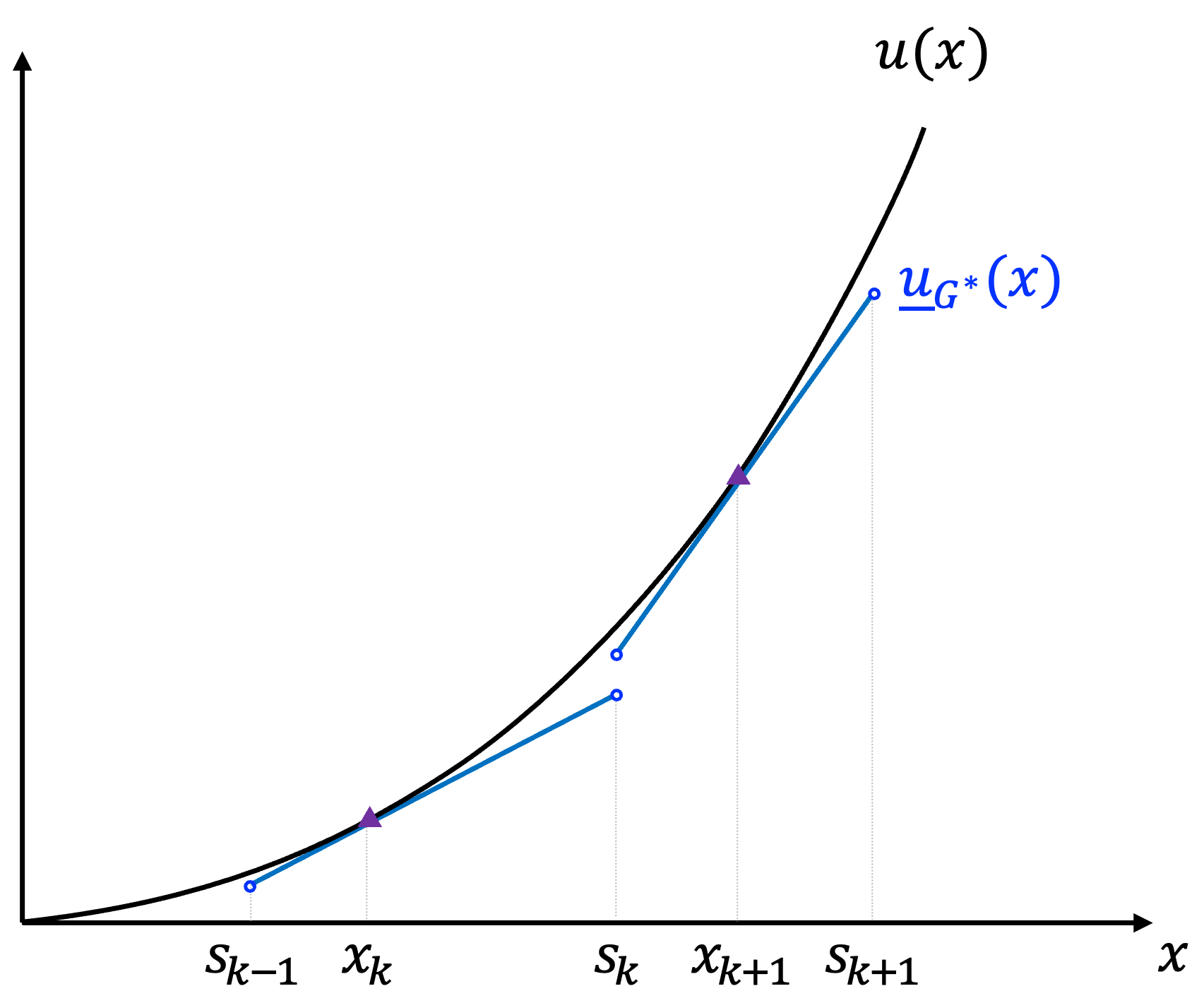}
		\end{subfigure}
		\hfil
		\begin{subfigure}[b]{0.4\textwidth}
			\centering
			\includegraphics[width=\textwidth]{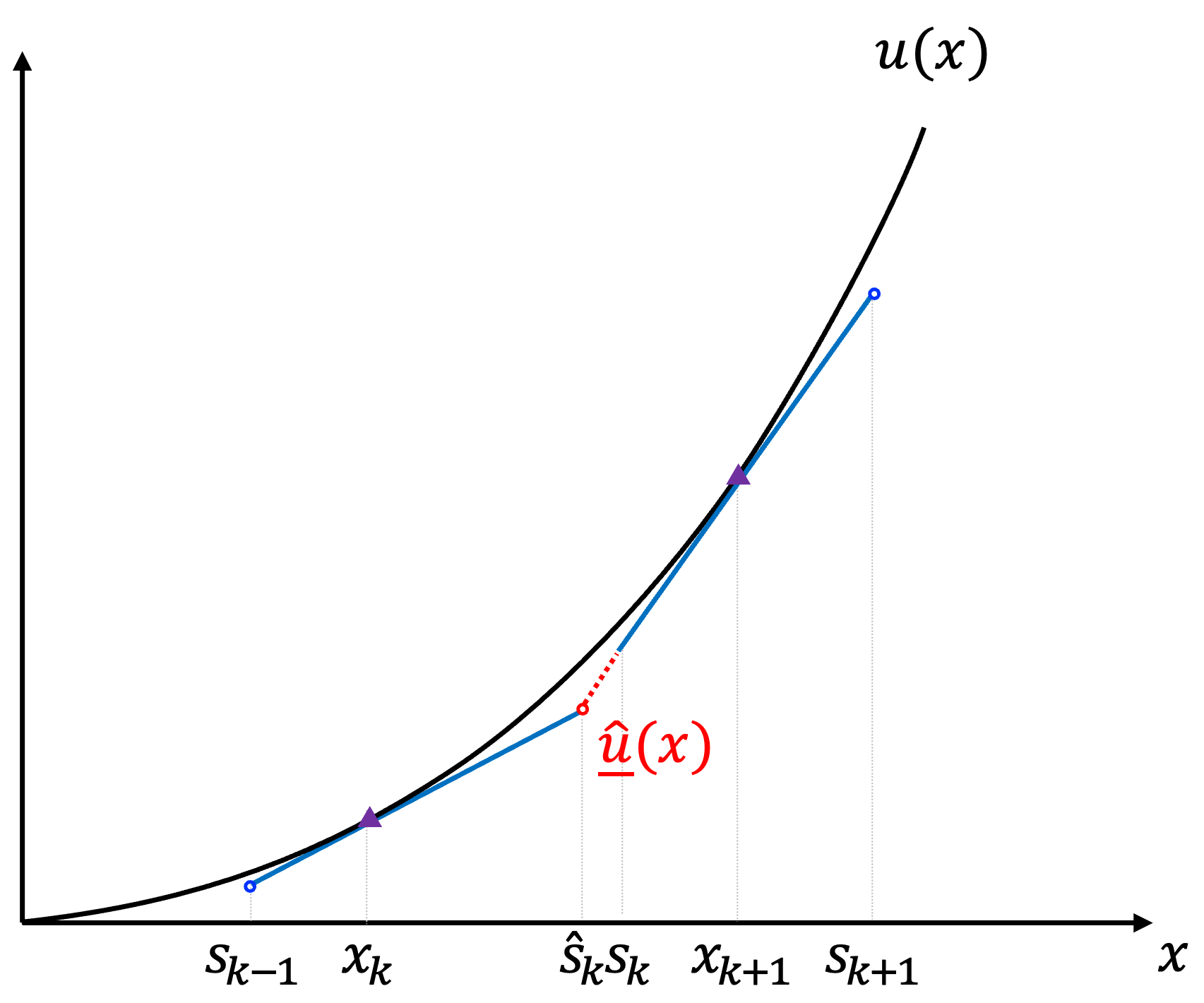}
		\end{subfigure}
		\caption{\small Suboptimality of discontinuous minorant function}
		\label{figure:proofdualexp}
	\end{figure}

	Now consider another interval structure $\hat{G}$, characterized by the partitional cutoffs $\{s_0,\dots,s_{k-1}, \hat{s}_k, s_{k+1},\dots,s_{N} \}$. The following inequalities show that this interval structure $\hat{G}$ generates higher expected payoff than $G^*$. 
	\begin{align*}
		&\int_0^1 u(x) \,\rmd G^\ast(x) =\int_0^1 \underline{u}_{G^\ast}(x) \,\rmd G^\ast(x)
		= \int_0^1 \underline{u}_{G^\ast}(x) \,\rmd F(x) \\
		& 
		\quad < \, \int_0^1 \hat{\underline u}(x) \,\rmd F(x)
		= \int_0^1 
		\hat{\underline u}(x)\,\rmd \hat{G}(x)
		\le \int_0^1 u(x) \, \rmd  \hat{G}(x). 
	\end{align*}
	The first equality follows from the fact that the information structure $G^\ast$ puts mass points of each subinterval $(s_{k-1},s_k)$ on the corresponding posterior mean $x_k$. The second equality comes from the linearity of $\underline{u}_{G^\ast}$ in each subinterval. Though the values of $\underline{u}_{G}(x)$ at the interval cutoffs are not assigned at this stage, this equality remains true because $F$ puts zero measure on this countable set.	The inequality on the second line follows because $\hat{\underline u}$ is everywhere above $\underline{u}_{G^\ast}$. The subsequent equality again comes from the piecewise linearity of  $\hat{\underline u}$. 	Finally the last inequality is due to the fact that $u$ is everywhere above $\hat{\underline u}$. This string of inequalities demonstrates that if $\underline{u}_{G^\ast}$ is not continuous at any cutoff point $s_k$, then there exists another information structure $\hat{G}$ that outperforms $G^\ast$---a contradiction. 
	
	We can therefore complete the minorant function at each $s_k$ by letting $\underline{u}_{G^\ast} (s_k)=\underline{u}_{G^\ast} (s_k^-)=\underline{u}_{G^\ast} (s_k^+) $. 
	The continuity of $\underline{u}_{G^\ast}$ implies that the completed minorant function under optimal information structure must be a convex, piecewise affine function.   Finally, if we regard $u$ as the ``integral distribution'' of $u'$, then $u$ bears a similar relationship to $\underline{u}_{G^*}$ as $I_F$ bears to $I_G$.   Consequently, Lemma \ref{lemma:conditional_expectation} 
implies that $s_k=\mathbb{E}_{u'}[\theta \,|\, \theta \in (x_{k},x_{k+1})]$ for $k=1,\ldots,N-1$. 
\end{proof}

Under an interval-partitional information structure, each signal $x_k$ is the posterior expectation of the state conditional on the interval $(s_{k-1},s_{k})$.
Theorem \ref{dualexpectation} states that under the \textit{optimal} interval-partitional information structure, each interval cutoff $s_k$ must be the expectation of  
the random variable that follows distribution $u'$, conditional on the interval $(x_{k},x_{k+1})$.\footnote{ 
We use $u'$ here as a shorthand to stand for the distribution function $(u'(\cdot)-u'(0))/(u'(1)-u'(0))$. }

To better understand the result, first note that due to the linear structure of payoff functions, every affine segment of the minorant function $\underline{u}_{G}$ on the subinterval $(s_{k-1},s_k)$ represents the designer's payoff when an action is taken upon the realization of the signal $x_k$.\footnote{
In the case of single-agent information acquisition problems, the piece of subgradient is exactly the designer's payoff at each state by the envelope theorem \citep{MilgromSegal2002}. When the designer and the receiver do not share perfectly aligned interests, due to the linear payoff structure, the subgradient generates the same expected payoff on each partitioned subinterval as the designer's actual payoff function.} 
Therefore, every two adjacent subgradients must intersect at the cutoff point $s_k$, for otherwise one can always construct a new information structure that produces a set of uniformly higher pieces of subgradients, yielding a higher payoff to the designer.
To put it differently, the designer should be indifferent between categorizing the cutoff state $s_k$ into either the upper subinterval or the lower subinterval under the optimal information structure. Figure \ref{figure:affineminorant} puts together the minorant function and the integral distribution of an interval partition. The minorant function $\underline u$ exhibits a geometric relationship with $u$ akin to the relationship between $I_G$ and $I_F$, which leads to the property of dual expectations. 

\begin{figure}[t]
	\centering
	\includegraphics[width=15cm]{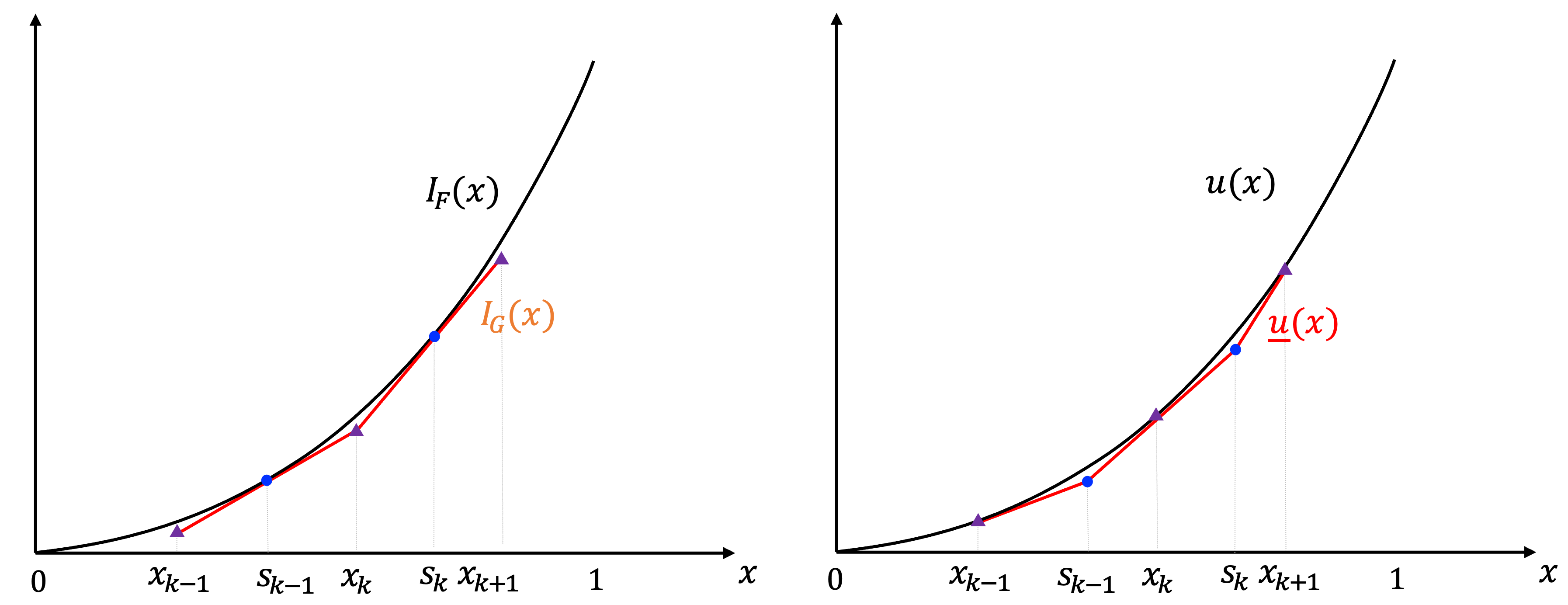}
	\caption{\small Dual expectations.}
	\label{figure:affineminorant}
\end{figure}

The characterization of (\ref{phi}) and (\ref{mu})  can extend to general value functions where $u'$ is a general signed measure with a well-behaved Radon-Nikodym derivative (with respect to Lebesgue measure) and the conditional ``expectations'' are calculated accordingly. We will provide more details in Sections \ref{section:s-shaped} and  \ref{generalsection}.

Theorem \ref{dualexpectation} gives the necessary condition for the optimal $N$-interval partition. The next result identifies environments under which the equation system has a unique solution, and therefore the necessary condition is also sufficient.
Logconcave priors are commonly used in information economics and are satisfied in many classes of probability distributions. In contrast, logconcavity of $u''$  concerns an endogenous object---the curvature of the value function.  In binary decision problems where the receiver has private information, the logconcavity of 
$u''$ is equivalent to the prior distribution of the receiver's private signal being logconcave, as illustrated in the purchase decision example (Example \ref{pd}). Later, we discuss a principal-agent problem (Example \ref{ex:pa}) and a coarse mechanism design problem (Example \ref{ex:cmd}), where the curvature distribution can also be logconcave.

\begin{proposition}\label{unique}
	If both $f$ and $u''$ are logconcave, the solution to the equation system (\ref{phi}) and (\ref{mu}) is unique and characterizes the optimal $N$-interval partition.
\end{proposition} 

The proof relies on the following property of logconcave densities. Because we will further exploit this property for subsequent results, we state it here for easier reference.  Obviously, a similar property as stated in Lemma \ref{keylemma} holds for the $\mu(\cdot,\cdot)$ function if $u''$ is logconcave.  
When both $f$ and $u''$ are logconcave, we refer to that as a \emph{logconcave environment}.

\begin{lemma}[\citet{MeaseNair2006,Szalay2012}]
	\label{keylemma}
	If $f$ is logconcave  and has full support, then for any $0 \leq a<b \leq 1$ and any $\eps_1, \eps_2 \ge 0$ such that $[a+\eps_1,b+\eps_2] \subseteq [0,1]$,
	\begin{equation*}
	\phi(a+\eps_1,b+\eps_2) \le \phi(a,b) + \max\{\eps_1,\eps_2\} ,
	\end{equation*}
with strict inequality if $\eps_1\ne \eps_2$. 
\end{lemma}
\begin{proof}[Proof of Proposition \ref{unique}.]
	The equation system (\ref{phi}) and (\ref{mu}) is a monotone mapping.  By Tarski's theorem, there exists a largest and a smallest fixed point (denoted $(x'_1,s'_1,\ldots,x'_N)$ and $(x''_1,s''_1,\ldots,x''_N)$, respectively).  If these two fixed points do not coincide, let $x'_j > x''_j$ where $j$ is the smallest index $j$ such that the element of the largest fixed point is strictly larger than the corresponding element of the smallest fixed point.  Because $s'_{j-1}=s''_{j-1}$, Lemma \ref{keylemma} implies that $s'_j - s''_j > x'_j-x''_j$.  Iterating this argument forward leads to $x'_N-x''_N > s'_{N-1}-s''_{N-1}$.  But $x'_N-x''_N=\phi(s'_{N-1},1)-\phi(s''_{N-1},1) < s'_{N-1}-s''_{N-1}$ by Lemma \ref{keylemma}, a contradiction.  A parallel argument applies if $s'_k > s''_k$, where $k$ is the smallest index such that the element in the largest fixed point is strictly larger.  Thus we have shown that the solution to the equation system is unique in a logconcave environment.  Because an optimal solution exists and there is only one candidate solution that satisfies the necessary condition for optimality, this candidate solution must be the optimal solution.
\end{proof}

When the value function is convex, the optimal $N$-interval partition characterized above also solves the general information design problem (\ref{auxiliaryprogram}) without restricting
to interval-partitional structures. 
The designer can gain a higher payoff by raising the informativeness of the experiment (i.e., by raising $I_G$) whenever $u''$ is positive.
Since it is always feasible to raise $I_G$ if $G$ is not interval-partitional, it is without loss of generality to focus on interval-partitional structures.
Moreover, the optimal experiment should always fully utilize the ``signal resources'' by setting $|\mathcal{K}_{I_{G}}| =N$. In other words, the ``information-maximizing'' designer always exhausts the signal budget.

\begin{proposition}
	\label{convexinterval}
	Suppose $u$ is convex. The optimal information structure $G$ is an interval-partitional structure with $|\mathcal{K}_{I_{G}}| =N$.  
\end{proposition}

The optimality of interval partition is also shown in \cite{Ivanov2021}. For completeness, we provide a simple proof in the Appendix based on the representation of integral distributions. Specifically, for any integral distribution $I_G$ that is not piecewise tangent to $I_F$ or that has less than $N$ pieces, we construct an alternative $I_{G^{\prime}}$ with $N$ pieces, 
induced by an interval-partitional information structure, such that $I_{G^{\prime}}$ lies everywhere weakly above $I_G$. This implies that the experiment $G$ is a mean-preserving contraction of $G^{\prime}$. The argument implies that the optimality of interval partitions holds for any convex value function $u$, even if $u$ is not strictly convex or not differentiable.

\subsection{Allocating scarce signal resources}
\label{allocation}

If a region of the state space is finely partitioned into more subintervals, then the experiment reveals finer details about that region of the state space as the receiver is more likely to distinguish across signal realizations. Given an interval partition, we say a subinterval receives \emph{closer scrutiny} than another if the width of this subinterval is smaller than that of another. Fixing the number of signals, if some region of the state space receives closer scrutiny, then some other region of the state space will receive less scrutiny.
A natural question is: Which part of the state space should receive closer scrutiny in the optimal experiment?

In the simplest case, both $f$ and $u''$ are uniform.  Then $\phi(a,b)=\mu(a,b)=(a+b)/2$.  Equations (\ref{phi}) and (\ref{mu}) imply that the state space $[0,1]$ is divided into $N$ equal-sized intervals with $s_k=k/N$ for $k=0,\ldots,N$; and the induced posterior means are evenly spaced, with $x_k=(2k-1)/(2N)$ for $k=1,\ldots,N$.  The optimal experiment gives every part of the state space equal scrutiny.

To develop more general lessons, denote the width of the $k$-th interval $(s_{k-1},s_k)$ by $w_k=s_k-s_{k-1}$ for $k=1,\ldots,N$, and the distance between two adjacent posterior means by $d_k=x_{k+1}-x_k$ for $k=1,\ldots,N-1$. Let 
$\{\Delta_i\}_{i=1}^{2N-1}$ be the interleaved sequence $\{w_1,d_1,w_2,\ldots,d_{N-1} ,w_N \}$,
where $\Delta_{2k-1}=w_k$ and $\Delta_{2k}=d_k$.
See Figure \ref{figure:interleavedwidth}.  We say that the sequence  $\{\Delta_i\}_{i=1}^{2N-1}$ is \emph{single-dipped} if it is decreasing then increasing or if it is monotone.

\begin{figure}[h]
	\centering
	\includegraphics[width=12cm]{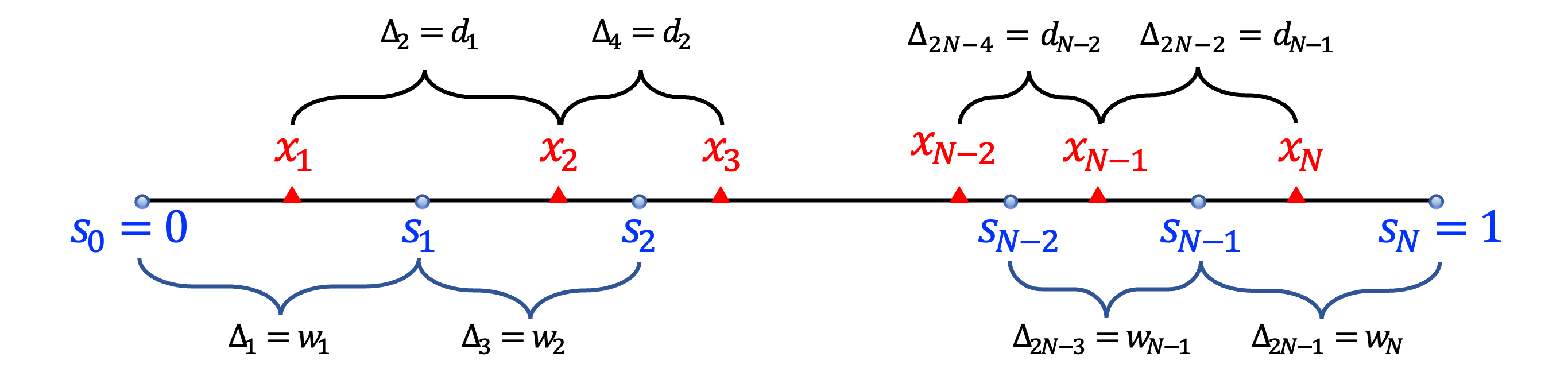}
	\caption{\small The interleaved sequence $\Delta_i$.}
	\label{figure:interleavedwidth}
\end{figure}

\begin{theorem}\label{intensityorder}
	In a logconcave environment, $\{\Delta_i\}_{i=1}^{2N-1}$ is single-dipped.
\end{theorem}

\begin{proof}
	It suffices to show that $\Delta_i \geq \Delta_{i-1}$ implies $\Delta_{i+1} \geq \Delta_i$ for any $i=2,\ldots,2N-2$. There are two cases to consider.

	Suppose $i$ is even, and $d_k \ge w_k$ but $w_{k+1} < d_k$ for some $k$.  Then
	\begin{align*}
		x_{k+1} =  \phi(s_{k},s_{k+1})
		= \phi (s_{k-1} + w_{k}, s_{k} + w_{k+1})
		< \phi(s_{k-1}+  d_{k},s_k+d_k )
		\leq x_{k}+d_k = x_{k+1}. 
	\end{align*}
The strict inequality obtains because $\phi$ is strictly monotone in each argument.  The second inequality comes from Lemma \ref{keylemma}.  This is a contradiction. Suppose	$i$ is odd, and $w_k \ge d_{k-1}$ but $d_k < w_k$.  Applying a similar argument above, we obtain 
	\begin{align*}
		s_k = \mu(x_{k},x_{k+1}) = \mu(x_{k-1} + d_{k-1}, x_{k}+d_k) < \mu(x_{k-1} + w_k, x_{k}+w_k) 	\le s_{k-1}+w_k =s_k, 
	\end{align*}
	which is also a contradiction.
\end{proof}

An immediate consequence of Theorem \ref{intensityorder} is that both the sequence of widths of subintervals, $\{w_k\}_{k=1}^N$, and  the sequence of distances between adjacent posterior means, $\{d_k\}_{k=1}^{N-1}$, are single-dipped in a logconcave environment. The part of the state space where $w_i$ attains the minimum is the region that receives the closest scrutiny in the optimal experiment. The result that the widths and distances are single-dipped implies there is a center of scrutiny---the subinterval with the minimum width. Moreover, the optimal experiment pays less and less scrutiny to states that are farther and farther away from this center.

Define the index $i^*$ such that $w_{i^*}=\min\{w_1,\ldots,w_N\}$.  The subinterval $(s_{i^*-1},s_{i^*})$ is the center of scrutiny of the optimal information structure. 
The location of this center generally reflects the designer's two concerns: (1) revealing finer information about the states that are more likely to occur (higher density $f$); and (2) focusing on states with higher local value of information (greater curvature $u''$).
Since the logconcavity of  $f$ and $u''$ implies that these functions are single-peaked \citep{DharmadhikariJoagdev1988}, the two ends of the state space have lower probability mass and smaller curvature. Intuitively, if we fix one of the concerns, such as assuming a uniform prior distribution, the scrutiny center is likely to be near the mode of the curvature.
The following result confirms this intuition and extends it to a more general environment where $f$ and $u''$ are single-peaked with possibly different modes, and they need not be logconcave.

\begin{proposition}
	\label{peak}
	\begin{enumerate}[(\roman*),nosep]
		
		\item If $u''$ is a constant and $f$ is single-peaked with mode $m_f$ in the $j$-th interval, then $i^* \in \{ j-1,j,j+1\}$.  
		
		\item If $f$ is a constant and $u''$ is single-peaked with mode $m_u$ in the $k$-th interval, then $i^* \in \{k-1,k,k+1\}$.
		
		\item If $f$ is single-peaked with mode $m_f$ in the $j$-th interval and $u''$ is single-peaked with mode $m_u$ in the $k$-th interval and $k \le j$, then $i^* \in \{ k-1,\ldots, j+1\}$.
	\end{enumerate}
\end{proposition}

\begin{mainexample}[Principal-agent model]\label{ex:pa}
	A risk-neutral agent with unknown productivity $\theta \in[0,1]$ chooses effort level $a \geq 0$ to maximize the expected wage net of effort cost $C(a)$. Output is given by $a \theta$, and the wage scheme is linear in output with a fixed piece rate $\beta \in(0,1)$. Thus, the agent's utility is $\beta a \theta-C(a)$.  
The prior distribution of $\theta$ is single-peaked at $m_f=0.5$. 
Let $a^*(x)$ be the agent's chosen level of effort when the posterior mean of $\theta$ is $x$.  The agent's interim utility is 
$u(x)=\beta a^*(x)x - C(a^*(x))$. Suppose the cost function takes the iso-elastic form, $C(a)=a^\gamma/\gamma$, with $\gamma > 1$.  Then, the agent's effort supply also takes the iso-elastic form, $a^*(x)=(\beta x)^{1/(\gamma-1)}$.  By the envelope theorem, $u'(x)=\beta a^*(x)$, and therefore
	\[
	u''(x)=  \frac{\beta^2}{\gamma-1} (\beta x)^{\frac{2-\gamma}{\gamma-1}}.
	\]
Now consider two cases. When $\gamma>2$, $u^{\prime \prime}(x)$ is strictly decreasing in $x$, so it is single-peaked with a mode at $m_u=0$. Intuitively, when the marginal cost of effort rises sharply, the agent's optimal effort becomes relatively unresponsive to expected productivity at higher values of $x$. As a result, the agent's payoff is not very sensitive to $x$ in the upper part of the productivity distribution, making it less valuable to finely differentiate high-productivity states. In this case, the optimal information structure for the agent will resemble Experiment B in Figure \ref{AvsB}, with a center of scrutiny near the lower part of the productivity distribution.  Note that in this case, $u''(x)$ is not logconcave, but the location of scrutiny center is still described by
Proposition \ref{peak}. On the other hand, 	
when $\gamma \in(1,2)$, $u''(x)$ is strictly increasing in $x$ and is also logconcave. In this case, the mode of curvature shifts to $m_u=1$. Here, the marginal cost of effort is less steep at the top, so effort is more sensitive to the change in productivity at the upper end. Consequently, the agent's utility is more sensitive to posterior means in the upper range, and the optimal information structure will resemble Experiment A in Figure \ref{AvsB},  with a center of scrutiny near the upper part of the productivity distribution.  
\end{mainexample}

\subsection{Comparative statics}
\label{section:compstat}
In this section, we present some comparative statics results as we vary the prior distribution $f$ and the curvature of value function $u''$.

\begin{definition}[Generalized Likelihood Ratio Order]\label{glr}
	\textnormal{A function $\hat f$ \textit{dominates} $f$ in likelihood ratio order (written as $\hat f \succeq_{lr} f$) if $\hat f(a) f(b) \le \hat f(b) f(a)$ for any $0\le a \le b \le 1$.} 
\end{definition}

If both $f$ and $\hat f$ are density functions with full support, then the above definition coincides with the standard likelihood ratio order, i.e., $\hat f(\cdot)/f(\cdot)$ is increasing. We adopt this multiplicative form because it can adapt 
to single-crossing functions, which is useful later on when we discuss comparative statics for S-shaped value functions (where the curvature function $u''$ is single-crossing).

\begin{proposition}
	\label{csforcutoffslr}
	In a logconcave environment, if either (a) the prior density changes from $f$ to $\hat f$, with $\hat f \succeq_{lr} f$; or (b) the value function changes from $u$ to $\hat{u}$, with $\hat u'' \succeq_{lr} u''$, then all the interval cutoff points $\{s_k\}_{k=1}^{N-1}$ and all the induced posterior means $\{x_k\}_{k=1}^N$ will increase.
\end{proposition}

\begin{proof}
	We prove the case where $f$ changes to $\hat{f}$. The proof for the change of value function is analogous. 
	A likelihood ratio increase in $f$ raises the conditional mean $\phi(a,b)$ for any given $a$ and $b$.  In a logconcave environment the optimal $(x_1,\ldots,x_N,s_1,\ldots,s_{N-1})$ is the unique fixed point of $\Gamma$, the mapping corresponding to the right-hand-side of the equation system (\ref{phi}) and (\ref{mu}).  Since a likelihood ratio increase from $f$ to $\hat f$ raises $\Gamma$, and $\Gamma$ is monotone, a standard result in monotone comparative statics \citep[Corollary 2.5.2]{Topkis1998}
establishes that any fixed point under $\hat f$ is larger than the unique (and hence smallest) fixed point under $f$. 
\end{proof}

Proposition \ref{csforcutoffslr} holds as long as the solution to the system of equations (\ref{phi}) and (\ref{mu}) is unique either in the original environment or in the new environment. The comparative statics result with regard to the change of prior distribution is intuitive.  A higher distribution (in the likelihood ratio order) means that the state is more likely to fall in the upper part of the state space. Naturally, the optimal experiment will induce more signals that reflect higher states.\footnote{
\cite{Tian2022} shows this part of the result in a more general environment where uniqueness is not required. \cite{Szalay2012}, \cite{ChenGordon2015}, \cite{DeimenSzalay2023} and \cite{SmithSorensenTian2021} make similar observations in the context of cheap talk models and social learning models.
}

Part (b) of Proposition \ref{csforcutoffslr} is more novel to the literature and deserves some additional discussion. For two value functions $u$ and $\hat u$, an increasing ratio of their corresponding curvatures $\hat u''/u''$ is equivalent to the condition that the marginal value function $\hat u'$ is \emph{more convex} than $u'$, i.e., there exists a strictly increasing and convex function $\psi$ such that $\hat u'(x)= \psi(u'(x))$.
This condition implies that the designer's value function under $u$ on average changes faster in the upper region. Therefore, the designer has incentives to learn higher states more precisely, as information is more valuable in higher states.

\examplepart{a (continued)}
We now revisit the purchase decision example through the lens of Proposition \ref{csforcutoffslr}. The buyer and seller differ in the relative value of information across different regions of the state space. 
The buyer's value function curvature is $h(x-p)$, 
which represents the ex ante probability that she is indifferent between making the purchase or not when the posterior mean is $x$. She values information most around the region where she is most likely to be indifferent (when consumer surplus from trade $x-p$ is near the mode of the cost shock). 
The seller, however, has a value function curvature $h'(x-p)$. He only cares about the probability of trade; hence information is most valuable to him when the probability that the buyer is on the margin between buying and not buying changes fastest.
When $h$ is an increasing function, the seller's value function is convex. Moreover, if $h$ is logconcave, then $
\hat u''/u''=p h'(x-p)/h(x-p)$ is decreasing in $x$. 
That is, the buyer's value of information increases faster 
with $x$ than does the seller's. 
Our result shows that the partitional cutoffs $\{\hat s_k\}_{k=1}^{N-1}$ and induced posterior means $\{\hat x_k\}_{k=1}^N$ under the buyer-optimal information structure are uniformly higher than the those under the seller-optimal information structure.

Next, we investigate how the optimal information structure changes when either $f$ or $u''$ becomes less variable. 

\begin{definition}[Generalized Uniform Variability Order] \label{guv}
	\textnormal{A function $\hat f$ is \textit{uniformly less variable} than $f$ (written as $f \succeq_{uv} \hat f$) if there exists some $p\in (0,1)$ such that $\hat f(a) f(b) \le \hat f(b) f(a)$ for all $0\le a \le b\le p$ and $\hat f(a) f(b) \ge \hat f(b) f(a)$ for all $p \le a \le b \le 1$. }   
\end{definition}

Again, we adopt this formulation to facilitate analysis in later sections where single-crossing functions are involved. 
When $\hat{f}$ and $f$ are density functions, the above definition reduces to $\hat f(\cdot)/f(\cdot)$ being unimodal.
If these two distributions are not (first-order) stochastically ordered, then $\hat f$ is said to be \emph{uniformly less variable than}  $f$ \citep{Whitt1985,ShakedShanthikumar2007}.\footnote{Uniform variability order does not require $f$ and $\hat{f}$ to have the same mean. For example, a normal distribution with a smaller variance is uniformly less variable than another normal distribution with a larger variance, regardless of the values of the respective means. 
}  
Intuitively, when either the density or the curvature becomes less variable, scrutiny tends to increase near the region where $f$ or $u''$ approaches the peak. The next result formalizes the reasoning.

\begin{proposition}
	\label{csforcutoffsuv}
	In a logconcave environment, if either (a) the prior density changes from $f$ to $\hat f$, with $f \succeq_{uv} \hat  f$; or (b) the value function changes from $u$ to $\hat u$, with $ u'' \succeq_{uv} \hat u''$,	then the new sequence of optimal signals and interval cutoffs, $\{\hat x_1,\hat s_1,\ldots, \hat s_{N-1},\hat x_N\}$, crosses the original sequence $\{x_1,s_1,\ldots,s_{N-1},x_N\}$ at most once and from above.
\end{proposition}

\begin{proof}
	We prove the proposition for (a); the proof for (b) is similar. Suppose $\hat f/f$ reaches a peak at $p \in (0,1)$, and let $k^\ast$  be the integer such that $s_{k^\ast} \le p < s_{k^\ast+1}$.	Define the sequence,
	\begin{equation*}
		\{\delta_i\}_{i=1}^{2N-1} :=\{x_1-\hat x_1, s_1-\hat s_1, \ldots,  s_{N-1}-\hat s_{N-1}, x_N-\hat x_N \}.
	\end{equation*}
To prove this proposition, it suffices to show that the two subsequences, $\{ \delta_i\}_{i=1}^{2k^\ast +1}$ and $\{\delta_i\}_{i=2k^\ast +1}^{2N-1}$, are both single-crossing from below. 

	Let $j$ be the first nonnegative term in the subsequence $\{ \delta_i\}_{i=1}^{2k^\ast +1}$.  We want to show that $\delta_j \ge 0$ implies $\delta_{j+1} \ge 0$. Suppose $j$ is odd and $j>1$; that is, $\delta_j = x_k - \hat x_k \ge 0$ and $\delta_{j-1}=s_{k-1}-\hat s_{k-1} < 0$  for some $k \in \{2,\ldots, k^*\}$. 
Let $\hat \phi(\cdot)$ be the conditional mean function under $\hat f$, we have 
	\begin{align*}
		x_k-\hat x_k  = \phi(s_{k-1},s_{k}) -\hat\phi(\hat s_{k-1},\hat s_{k})  
		&\le \phi(s_{k-1},s_{k}) -\phi(\hat s_{k-1},\hat s_{k}) \\
		&< \phi(\hat s_{k-1},s_{k}) -\phi(\hat s_{k-1},\hat s_k)
		\leq \max \{ 0, s_{k}-\hat s_{k}\},
	\end{align*}
where the first inequality follows because $\hat f(\cdot)/f(\cdot)$ is increasing on $[0,p]$, and the second inequality follows because $\phi(\cdot,s_k)$ is strictly increasing. Since $\delta_j=x_k-\hat x_k$ is nonnegative, the above sequence of inequalities implies that 
$\delta_{j+1}=s_k-\hat s_k> 0$ and  $\delta_{j+1}=s_k-\hat s_k>x_k-\hat x_k =\delta_j$. Carrying this argument forward,
	\begin{equation*}
		s_k-\hat s_k  = \mu(x_{k},x_{k+1}) -\mu(\hat x_{k},\hat x_{k+1})  \leq \max \{ x_{k}-\hat x_{k}, \max\{x_{k+1}-\hat x_{k+1},0\}\}.
	\end{equation*}
Since $s_k-\hat s_k > x_k-\hat x_k$, the above inequality implies $\delta_{j+2}=x_{k+1}-\hat x_{k+1}\ge \delta_{j+1}> 0$. By induction,  $\{\delta_{j},\ldots,\delta_{2k^*+1}\}$ is positive and increasing. 
	
	Suppose $j=1$, then $x_1 \ge \hat x_1$ and $s_{0}=\hat s_{0}=0$.  If $x_1 > \hat x_1$,  we have
	\begin{align*}
		x_1-\hat x_1  = \phi(0,s_{1}) -\hat\phi(0,\hat s_{1})  
		&\le \phi(0,s_{1}) -\phi(0,\hat s_{1})<\max \{ 0, s_{1}-\hat s_{1}\}.
	\end{align*}
    
Repeating the same argument,  $\{\delta_{1},\ldots,\delta_{2k^*+1}\}$  is positive and increasing. If $x_1 = \hat x_1$, then $s_1 \ge \hat s_1$. If the inequality is strict, we can resort to the same argument as above. Otherwise, if all inequalities hold as equality while we proceed, the subsequence $\{ \delta_1, \dots,\delta_{2k^\ast+1} \}$ are all zero. 
	
	Similar reasoning applies to the case where $j$ is even. Thus, we establish that $\{ \delta_i\}_{i=1}^{2k^\ast +1}$ is single-crossing from below.
	
	For the subsequence $\{ \delta_i\}_{i=2k^\ast +1}^{2N-1}$, the single-crossing property is equivalent to the single-crossing property of  the following:
	\begin{equation*}
		\{\hat x_N-x_N, \hat s_{N-1}-s_{N-1}, \ldots,	\hat s_{k^\ast+1}-s_{k^\ast+1},\hat x_{k^\ast+1}-x_{k^\ast+1} \}.
	\end{equation*}
Note that $f(\cdot)/\hat f(\cdot)$ is increasing on $[p,1]$. We apply a symmetric argument to prove the sequence listed above is single-crossing from below.	Combining the single-crossing property of $\{ \delta_i\}_{i=1}^{2k^\ast +1}$ and $\{\delta_i\}_{i=2k^\ast +1}^{2N-1}$  establishes the desired result.
\end{proof}

\begin{figure}[t]
	\centering
	\includegraphics[width=11cm]{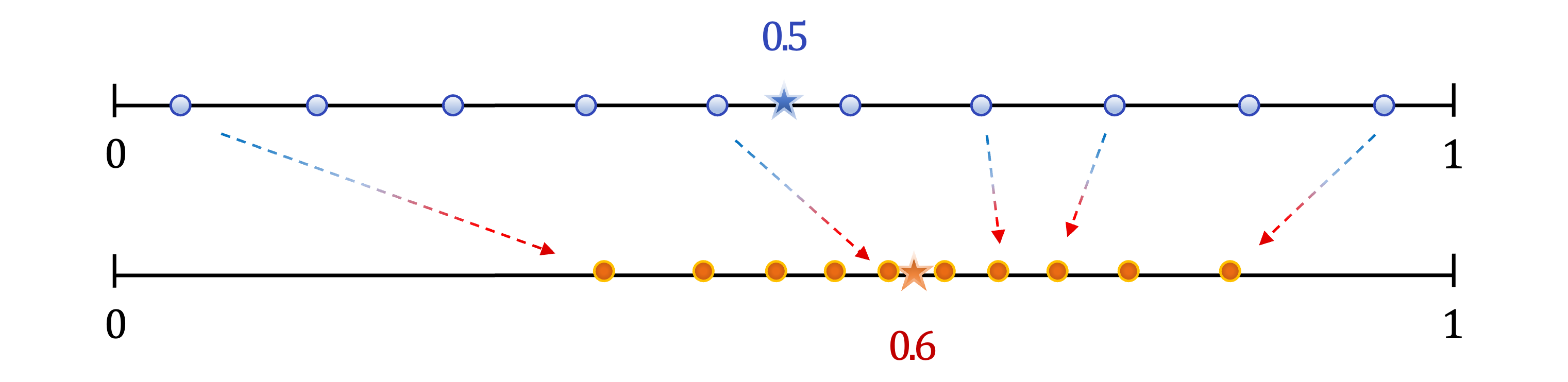}
	\caption{
		Comparative statics under uniform variability order.
	}
	\label{fig:dispersedinfostructure}
\end{figure}

An implication of Proposition \ref{csforcutoffsuv} is that $\{\hat s_k\}_{k=1}^{N-1}$ crosses $\{s_k\}_{k=1}^{N-1}$ at most once and from above. If they indeed cross, there exists a $n^*$ such that all cutoffs to the left of $s_{n^*}$ shift up and all cutoffs to the right of $s_{n^*}$ shift down.
See Figure \ref{fig:dispersedinfostructure} for an illustration. This figure shows that a uniformly less variable density will cause the sequence of optimal partitional cutoffs to be more compressed.
Moreover, because $s_1$ increases and $s_{N-1}$ decreases, the widths of the intervals other than the outermost ones must become narrower on average. In this sense, the designer is giving more scrutiny to the middle part of the state distribution when either $f$ or $u''$ becomes uniformly less variable. 
Intuitively, a less variable $f$ or $u''$ means either the likelihood or the relative information value is now more compressed toward the mode. Naturally the designer should respond by compressing the interval partition to be finer
around the mode as well.

\section{Coarse Mechanism Design} 
\label{nonlinearpricing}

Our characterization of dual expectations has implications beyond persuasion models. 
In particular, we can adapt the minorant function approach to study mechanism design problems where the agent's payoff is linear in 
her private type and where the contract space is constrained to be finite. We illustrate this with a simple model of nonlinear pricing with a finite menu.\footnote{
\cite{Wong2014} uses the first-order approach to derive results for this problem. He mainly focuses on the marginal benefit from increasing the size of the menu, and discusses the approximation property when $N$ grows large. \cite{BergemannYehZhang2021} study a linear-quadratic version of this model using the quantization approach. 
\cite{BergemannHeumannMorris2025} 
show the optimality of endogenously coarse menu when the seller can jointly design the information structure and selling mechanism.}   
Our main objective is to demonstrate the transferability of our analytical framework to an analogous coarse mechanism design problem.

Consider the classic nonlinear pricing model of \cite{MussaRosen1978}. A seller supplies products to a continuum of buyers with unit demand. Each buyer has private information about 
her preference type $\theta \in [0,1]$, distributed according to $F$ with a continuous density function $f$. The seller designs a menu that specifies a set of quality and price pairs $\{(q_i,p_i)\}_{i \in \Sigma}$. A buyer with type $\theta$ who chooses a specific pair $(q,p)$ from the menu receives the payoff $ \theta v(q) - p $. Seller's production cost of supplying quality $q$ is $c(q)$.
For simplicity, we assume $v$ and $c$ are twice differentiable with $v'>0$, $v''<0$, $c'>0$, $c''>0$, and $v(0)=c(0)=0$.  
Additionally, we impose the constraint that the seller can only offer menus with a finite number of options, i.e., 
$|\Sigma| \le N$.
Throughout our discussion, we focus on the profit-maximizing menu for the seller; the analysis can be readily adapted to study welfare maximization by a planner.

By the revelation principle, we can focus on direct mechanisms $\{q(\theta),p(\theta)\}$, where $q(\cdot)$ and $p(\cdot)$ are step functions because the menu is constrained to be finite.
In fact, we can further restrict our attention to (weakly) increasing allocation rule $q(\theta)$, as any such rule is implementable by adjusting the transfers $p(\theta)$.
For each incentive compatible allocation $q(\theta)$, the seller's profit takes the well-known formula,
\begin{equation}
	\label{mechanism}
	\int_0^1 \left[ v\left( q\left(\theta\right)\right) \left(\theta- \frac{1-F\left(\theta\right)}{f\left(\theta \right)}\right)  -c\left(q(\theta)\right)\right] \, \rmd F\left( \theta\right). 
\end{equation}
Define $\varphi:= \theta -(1-F(\theta))/f(\theta)$ to be the virtual valuation, and let $H$ be its associated distribution on $[0,1]$. Throughout the section, we assume that virtual valuation is strictly increasing in type, and $H$ admits 
a continuous density function $h$ with full support. 
Henceforth, we directly work with virtual valuation.

With slight abuse of notation, let $q(\varphi)$ represent the quality allocation to virtual type $\varphi$ in a finite menu, and define 
$\underline{\pi}(\varphi) := \varphi v(q(\varphi))-c(q(\varphi))$ correspondingly.
Formally, the coarse mechanism design problem is the following:
\begin{align}
	\label{CMD}\tag{CMD}
	&\max_{q(\cdot)} \quad  \int_0^1 \underline{\pi}(\varphi)\,\mathrm{d}H(\varphi)\\
	\qquad\qquad\qquad &\text{s.t.} \quad q(\cdot) 
	\text{ is nondecreasing},  \tag{Monotonicity} \label{Monotonicity}\\
	&\phantom{\text{s.t.}} \quad  q([0,1])  \text{ has at most $N+1$ elements}. \tag{Finite Menu} \label{finitemenu} 
\end{align} 

The program is essentially to pick a monotone, finite allocation rule $q(\varphi)$ to maximize the expected virtual surplus.
Exclusion of consumers is captured by $q(\varphi)=0$. Therefore 
an $N$-option menu is equivalent to an $N+1$-valued allocation rule.
Since the allocation rule has to be monotone, the optimal menu partitions the (virtual) type space into $N+1$ subintervals with a sequence of cutoff types, 
$0=s_0<s_1< \ldots < s_{N+1}=1$. 
Within each type segment 
$(s_{k-1},s_{k})$  where $k=2,\ldots,N+1$,
the allocation assigns $q_k$ uniformly.  Consumers with virtual types in $[0,s_1)$ are excluded and are assigned $q_1=0$.

To leverage the dual expectations property in this coarse mechanism design problem, let
\begin{equation}\label{md}
	\pi(\varphi) :=  \max_{q\ge 0 } \ \varphi v (q)-c(q)  
\end{equation}
be the pointwise maximal virtual surplus on $\varphi\in [0,1]$. By the smoothness of $v$ and $c$, $\pi$ is twice-continuously differentiable. 
Let $q^*(\varphi)$ be the unconstrained optimal allocation to (\ref{md}).  
Because the pointwise profit is linear in the consumer's virtual type, $\pi(\varphi)$ is increasing and convex on $[0,1]$. 
Moreover, 
for any $0 \le a < b \le 1$, 
we define\footnote{
	Again by a scaling argument and by the convexity of $\pi$, it is without loss of generality to normalize its derivative $\pi'$ to be a cumulative distribution function. }
\begin{align*}
	\phi^m(a,b) & := \bbE_{H}\left[ t\, \middle \vert\, t \in \left( a,b \right)\right],\\
	\mu^m(a,b) & := \bbE_{\pi'}\left[ t\, \middle \vert \, t \in \left(a,b \right)\right] . 
\end{align*}

\begin{proposition}
	The solution to program (\ref{CMD}), characterized by $\{ s_k\}_{k=0}^{N+1}$ and $\{q_k \}_{k=1}^{N+1}$, 
	must satisfy:
	\begin{align}
		x_k &= \phi^m(s_{k-1},s_{k}) \qquad \text{for } k=2,\ldots,N+1; \label{CMD-signal} \tag{CE-$H$}
		\\
		s_k &= \mu^m(x_k,x_{k+1}) \quad \ \ \, \text{for } k=1,\ldots,N;
		\label{CMD-cutoff} \tag{CE-$\pi'$}
		\\
		q_k &= q^\ast \left( x_k \right)  \qquad  \qquad   \text{for } k=2,\ldots,N+1. \label{CMD-allocation} \tag{Quality}
	\end{align}
	with $x_1=0$ and $q_1=0$. 
\end{proposition}

\begin{proof}
	By definition, $\underline{\pi}(\varphi)$  is weakly below $\pi(\varphi)$ for any $\varphi \in [0,1]$.  Our goal is to show that the optimal $q(\varphi)$ induces a virtual surplus function $\underline{\pi}(\varphi)$  that shares the same structure as the minorant function $\underline{u}$ identified in the corresponding coarse information design problem, except for the first segment $(0,s_1)$ where there is no trade.	Any incentive-compatible $N+1$-quality allocation rule $q(\varphi)$ would induce some $\underline{\pi}$ which is piecewise affine with $N+1$ pieces, including a horizontal piece in the first segment. We show that, under the optimal $q(\varphi)$, the corresponding $\underline{\pi}$  satisfies the following two properties: (a) it is continuous at $s_k$ for all $k=1,\ldots,N$ and thereby is also convex; and (b) on each type segment $\left(s_{k-1},s_{k}\right)$ with $k=2,\ldots,N+1$, it is tangent to $\pi$ at the conditional expectation of the corresponding segment $x_k$. 
	
	For part (a), suppose the surplus function $\underline{\pi}$ induced by an allocation $q(\cdot)$ is discontinuous at some cutoff point $s_k$, as depicted in the left panel of Figure \ref{efficientmenu}. Then we can construct another incentive-compatible allocation $\hat q(\cdot)$, which supplies higher quality to consumers in $(\hat s_k, s_{k})$ while the allocations to other consumers are unchanged. The corresponding $\underline{\hat {\pi}}$  function (shown by the blue dotted line for $\varphi \in (\hat s_k,s_k)$ in Figure \ref{efficientmenu}) for this modified allocation is everywhere higher than $\underline{\pi}$ for the original allocation. This would yield a strictly higher profit, hence a contradiction.	The convexity of the optimal $\underline{\pi}$ then follows directly from continuity and from the fact that the slope of $\underline{\pi}$ is equal to $v(q(\varphi))$, where $q(\cdot)$ must be nondecreasing for incentive compatibility.
	
	\begin{figure}[t]
		\begin{minipage}{0.45\textwidth}
			\centering
			\includegraphics[width=\textwidth]{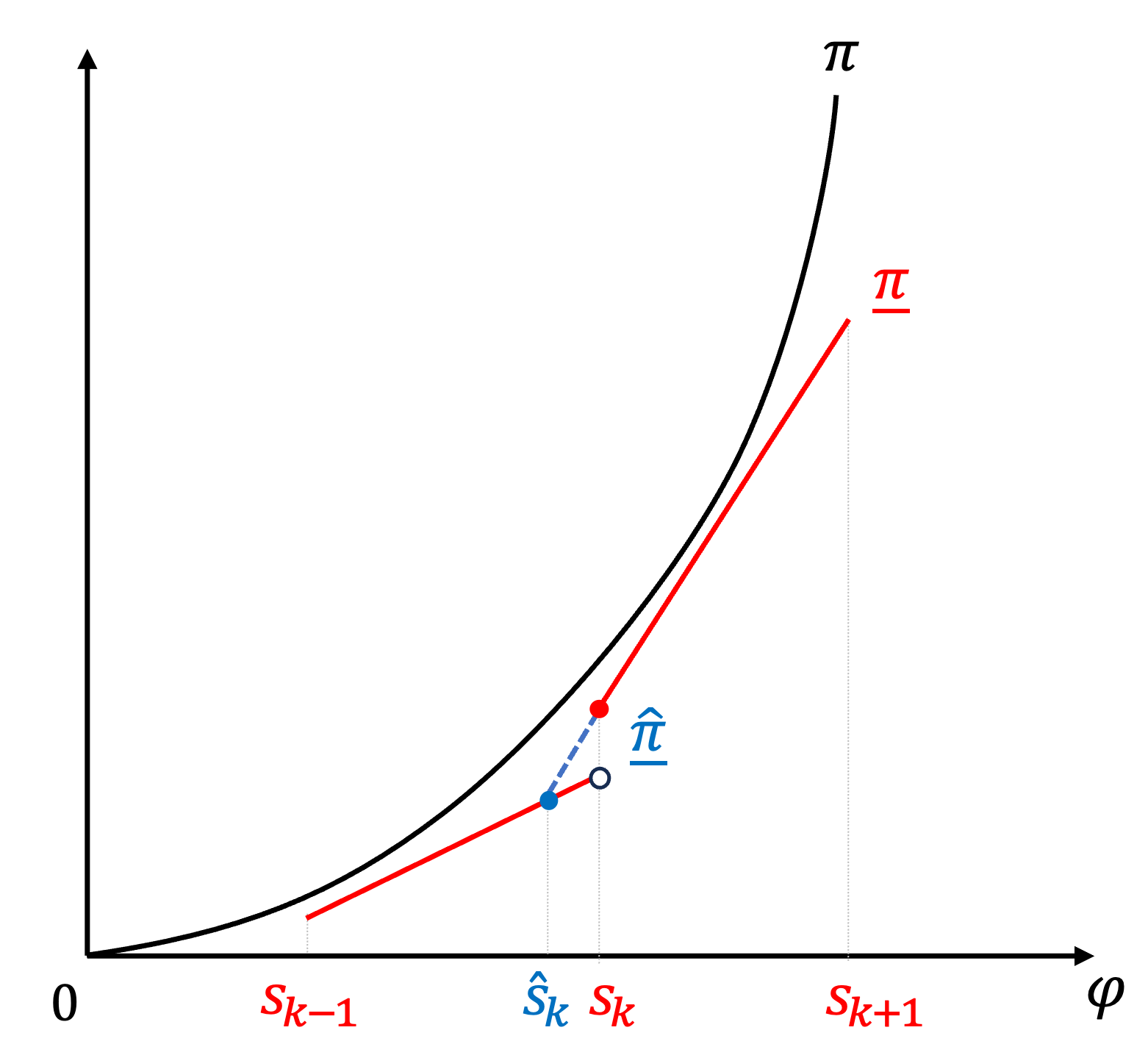} 
		\end{minipage}\hfill
		\begin{minipage}{0.45\textwidth}
			\centering
			\includegraphics[width=\textwidth]{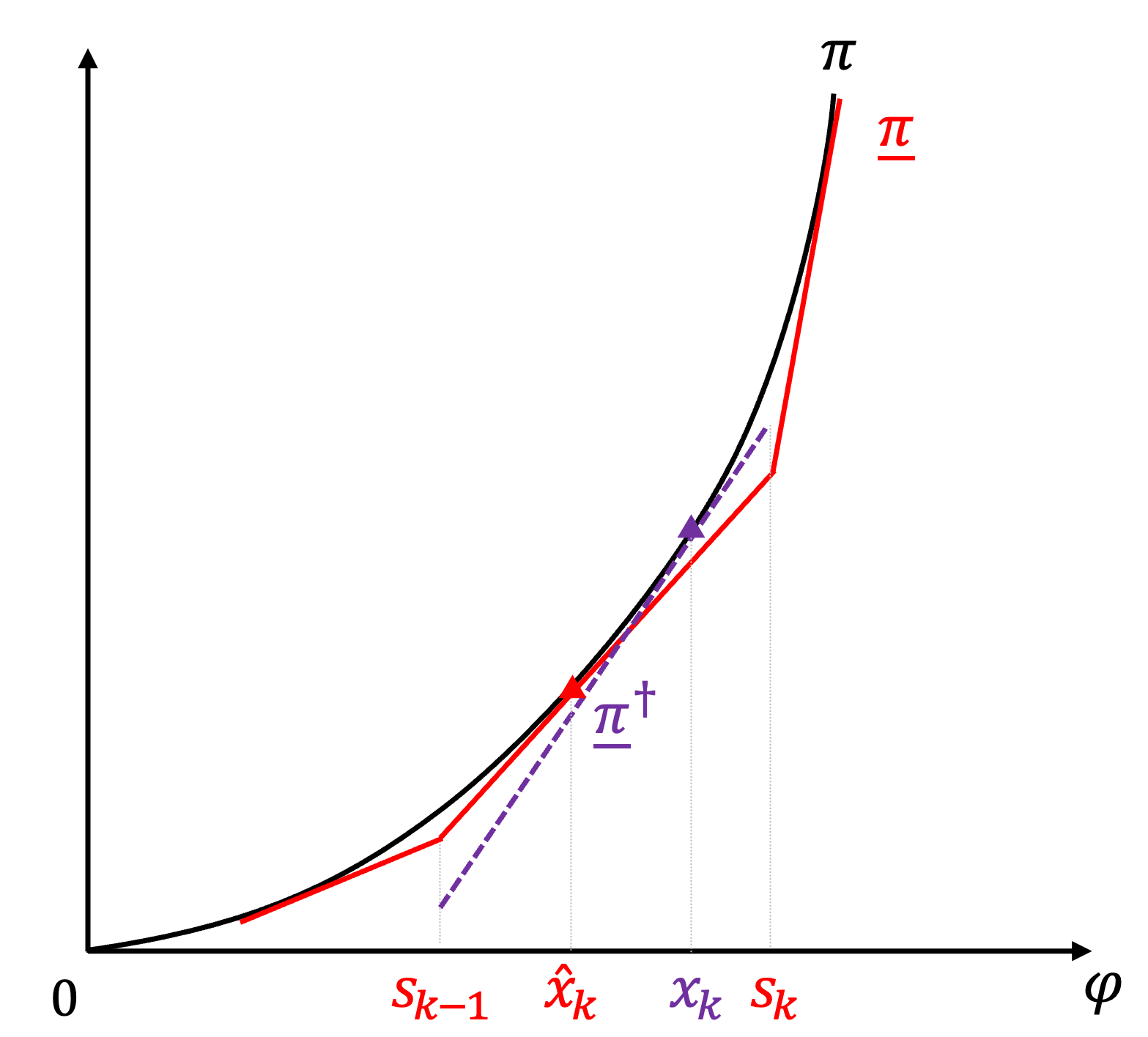} 
		\end{minipage}
		\caption{\small Properties of optimal $\underline{\pi}$:
it is continuous (left panel) and is tangent to $\pi$ at $x_k=\phi^m(s_{k-1},s_k)$ (right panel).
}
		\label{efficientmenu}
	\end{figure}
	
	For part (b), first observe that each piece of $\underline{\pi}$ must be tangent to $\pi$ at some interior point, for otherwise the seller can adjust the allocation to raise $\underline{\pi}$ and increase profit. It remains to be proven that the optimal allocation for the segment $(s_{k-1},s_k)$  must equal the optimal allocation for the average virtual type of that segment, i.e., $q_k=q^*(x_k)$.	Suppose not, and instead $\underline{\pi}$ is tangent to $\pi$ at some point $\hat{x}_k\ne x_k$. Then $\underline{\pi} (\varphi)= \pi(\hat x_k) +  \pi'(\hat x_k) (\varphi -\hat x_k )$ (shown as the red line in the right panel of Figure \ref{efficientmenu}). We can construct another allocation $q^\dagger(\cdot)$ such that $q_k^\dagger= q^\ast(x_k)$ with the corresponding surplus function $\underline{\pi}^\dagger$ tangent to $\pi$ at $x_k$ (shown as the purple dotted piece), and is otherwise equal to the original $q$.  Then,
	\begin{align*}
		\int_{s_{k-1}}^{s_k} \underline{\pi}(\varphi) \,\mathrm{d}H(\varphi) &= \left[H(s_k)-H(s_{k-1})\right]\left[\pi(\hat{x}_k)+ \pi'(\hat{x}_k)(x_k-\hat{x}_k)\right]\\
		&< \left[H(s_k)-H(s_{k-1})\right]\pi(x_k)=\int_{s_{k-1}}^{s_k} \underline{\pi}^\dagger \,\mathrm{d}H(\varphi),
	\end{align*}
where the inequality follows from the strict convexity of $\pi$.
This contradicts the optimality of the allocation corresponding to the original $\underline\pi$ function. Therefore $q_k = q^\ast \left( x_k \right) $ must hold.
	
	Consequently, the virtual surplus function $\underline{\pi}$ under the optimal quality allocation must be a convex and piecewise affine function. We can again use Lemma \ref{lemma:conditional_expectation} to conclude that (\ref{CMD-cutoff}) must hold.
\end{proof}

Our analysis implies that the optimal 
$\underline{\pi}$  exhibits the same property as the minorant function that we introduce before. 
For $k=2,\ldots,N+1$, the optimal menu allocates goods with quality $q_k=q^*(x_k)$ to consumers with virtual types belonging to the segment $(s_{k-1},s_k)$, where $x_k$ is the expected type of this segment.   Consumers with virtual types $\varphi \in (0,s_1)$, together with those with negative $\varphi$, are excluded and receive the outside option.  
Moreover, for $k=1,\ldots,N$, since the interval cutoff $s_k$ is a kink-point of $\underline{\pi}$, we have $s_k$ as the expected value of the interval $(x_k, x_{k+1})$ under the measure $\pi'$. Hence, the dual expectations characterization is still valid for the finite-menu nonlinear pricing problem, and the remaining results in Section \ref{section:Convex} continue to hold (where the prior distribution is interpreted as the distribution of virtual type). 

\begin{mainexample}[Finite menu and consumer exclusion]\label{ex:cmd}
	Suppose $v(q)= q^\beta $ and $c(q)=\gamma q$ for 
	$\beta\in (\frac{1}{2},1)$.  The unconstrained optimal allocation for the seller is $q^*(\varphi)=(\beta\varphi/\gamma)^{1/(1-\beta)}$, and the ``value function'' is
$\pi(\varphi)=(1-\beta)(\gamma/\beta)(\beta\varphi/\gamma)^{1/(1-\beta)}$.  Let the value of the good change to $\hat v(q)=q^{\hat\beta}$, with $\hat \beta > \beta$. Note that $\hat v(q) > v(q)$ for all $q > 1$.\footnote{ 
The seller's optimal menu will  offer $q_k > 1$ for all $k>1$ when the marginal cost $\gamma$ is sufficiently low. }  
The value function $\hat \pi (\varphi)$ corresponding to the higher valuation of quality satisfies the property that $\hat \pi''/\pi''$
is increasing.  If the density of virtual valuation is logconcave on $[0,1]$, this setting is a logconcave environment.  By the same logic leading to Proposition \ref{csforcutoffslr}, all the interval cutoff points will move to the right when the benefit increases from $v(q)$ to $\hat v(q)$.  In particular, we have $\hat s_1 > s_1$.  This means that the seller's optimal menu will exclude a larger set of consumers, despite the fact that the high-quality good has become more valuable to all consumers. Intuitively, an increase in $\beta$ not only implies higher valuation for quality, but also leads to greater importance to screen the high virtual value buyers. When the screening device is constrained, the optimal finite menu naturally involves more discrimination towards high-end customers, leading to the exclusion of more low-end customers. This is a novel feature for finite menu design, as the coverage of customers in an unconstrained problem only depends on the distribution of virtual value and will not change when $v$ changes.
\end{mainexample}

\section{S-shaped Value Functions}
\label{section:s-shaped}

The analysis so far is based on the premise that the value function $u$ is convex, under which the designer wishes to transmit as much information as possible. 
In this section, we extend our analysis to S-shaped value functions, where the designer faces the tension between providing information on the convex region and suppressing it on the concave region.

Formally, a function $u$ is \emph{S-shaped} if there exists some inflection point $t \in (0,1)$ such that $u$ is convex on $(0,t)$ and concave on $ (t,1)$. In other words, $u''$ is single-crossing from above.\footnote{ 
With obvious modification, the analysis in this section also extends to the case where $u''$ is single-crossing from below.}
S-shaped value functions have been extensively studied in persuasion literature because they capture a range of economic applications. The unconstrained optimal information structure $G^\ast$ features ``upper-censorship,'' namely the optimal experiment without the finiteness constraint will reveal full information on $[0,s^*)$  and coarsen information by pooling all states in $(s^*,1]$ for some $s^* < t$ \citep{KolotilinMylovanovZapechelnyuk2022}.
In this section, we rule out the uninteresting case where the unconstrained solution reveals no information.

Under S-shaped value functions, the dual expectation characterization, the single-dipped property, as well as the comparative statics results, remain true when the relevant conditions are properly modified. In particular, when the value function $u$ is not convex, $u''$ can take negative values and may not be a valid density.  
Nevertheless, it can be interpreted as the ``density'' corresponding to a \emph{signed measure}.  For any $0\le a<b\le 1$, define its \emph{barycenter} under signed measure $u'$ by 
\begin{equation}
	\mu(a,b) :=	\frac{\int_a^b  \theta u''(\theta) \, \rmd \theta}{\int_a^b  u''(\theta) \, \rmd \theta}.
	\label{barycenter}
\end{equation} 
whenever $\int_a^b  u''(\theta) \, \rmd \theta \neq 0 $. 
We can therefore interpret $\mu(a,b)$ in equation (\ref{mu}) as the barycenter of the set $(a,b)$ under signed measure $u'$.  

\begin{proposition}\label{dual-S}
	Suppose $u$ is S-shaped. The optimal information structure $G$ is an interval-partitional structure with $|\mathcal{K}_{I_{G}}| =N$.  It satisfies the equation system (\ref{phi}) and (\ref{mu}).  Moreover, the optimal information structure $G$ is less informative than the unconstrained one $G^*$; in particular,
    $s_{N-1}< s^\ast$. 
\end{proposition}

\begin{proof}[Sketch of proof.]
	The proof of the optimality of interval-partitional structure follows a similar logic as the proof of Proposition \ref{convexinterval}: raising $I_G$ increases the value of the objective function (\ref{auxiliaryprogram}) whenever $u'' > 0$.  Suppose $I_G$ corresponds to the optimal information structure
and there is a segment of $I_G$ that is not tangent to $I_F$ (so that $G$ is not interval-partitional). If $u''$ is not entirely positive on this segment (otherwise, the argument is similar),
then we consider an alternative information structure $I_{G'}$ which is a clock-wise rotation of $I_G$, with $I_{G'} > I_G$ to the left of the inflection point $t$ on the same segment and $I_{G'} < I_G$ to the right of the inflection point on the same segment.  This alternative information structure will increase informativeness locally in the convex region and reduce informativeness in the concave region and raise the value of (\ref{auxiliaryprogram}).
	
	The necessity of (\ref{phi}) and (\ref{mu}) follows from Propositions \ref{generaloptimal} and \ref{generalminorant}, which will be presented in Section \ref{generalsection} for general value functions.  The proof that    $s_{N-1} < s^*$ is provided in the online appendix.  
\end{proof}

In the information design problems without finiteness constraint, the optimal censorship threshold $s^*$ lies in the convex region, and it solves the designer's local trade-off between providing full information around the marginal state $s^\ast$ in the convex region and pooling it into the concave region. With the finiteness constraint, suppose the interval cutoff of the last interval is $s_{N-1}=s^*$, then the marginal gain from information provision around the same state diminishes since now $s^\ast$ cannot be perfectly revealed, but is instead pooled with the lower interval $(s_{N-2},s^\ast)$. As a result, to equalize the margin, the designer optimally pools more of the convex region with the concave region to diminish the marginal gains from censorship at the upper end. 
The optimal cutoff point under $G$ satisfies $s_{N-1}< s^\ast$.  
This implies that $G$ is less Blackwell-informative than $G^*$.   

The main challenge in generalizing the remaining results from the case of convex value functions arises because the key logconcavity property (Lemma \ref{keylemma}) no longer applies when the curvature $u''$ is
single-crossing. Nevertheless, 
a one-sided version of Lemma \ref{keylemma} still holds.

\begin{lemma} \label{keylemma-S}
	Suppose $u''$ is logconcave on $(0,t)$.  Then for any $0\le a< t<b \le 1$ and $\varepsilon>0$ such that $u'(b)-u'(a+\varepsilon)>0$, we have 
	\[
	\mu(a+\varepsilon, b) \leq \mu(a, b)+\varepsilon.
	\]   
\end{lemma}

This weaker condition, together with Lemma \ref{keylemma}, turns out to be sufficient for the remaining results because, under the optimal interval-partitional information structure, all interval cutoffs $\left\{s_k\right\}_{k=1}^{N-1}$ and 
the posterior means $\left\{x_k\right\}_{k=1}^{N-1}$ 
remain within the convex region of the value function. The only complication is 
that the last posterior mean $x_N$ may lie in 
either the convex region or in the concave region, and therefore requires a case-by-case verification. To maintain focus of the main text, we put the proof of the following proposition in an
online appendix.

\begin{proposition}\label{intensityorder-S}
	Suppose $u$ is S-shaped, $f$ is logconcave, and $u''$ is logconcave on $(0,t)$.  
	\begin{enumerate}
		\item[(i)] $\{\Delta_i\}_{i=1}^{2N-1}$ is single-dipped.
		\item[(ii)] The solution to the equation system (\ref{phi}) and (\ref{mu}) with $s_{N-1}<t$ is unique.
		\item[(iii)] If (a) the density changes from $f$ to $\hat f$,  with $\hat f \succeq_{lr} f$; or (b) the value function changes from $u$ to $\hat u$, with $\hat u'' \succeq_{lr} u''$  and $\hat u''$ is single-crossing at $\hat t$ and logconcave on $(0,\hat t)$, then all interval cutoffs and all induced posterior means will increase. 
		\item[(iv)] If (a) the density changes from $f$ to $\hat f$, with $ f \succeq_{uv} \hat f$; or (b) the value function changes from $u$ to $\hat u$, with $ u'' \succeq_{uv}\hat u''$
where
$\hat u''$ is single-crossing at $\hat t$ and logconcave on $(0,\hat t)$, then the new sequence of optimal posterior means and interval cutoffs, $\{\hat x_1,\hat s_1,\ldots,\hat s_{N-1},\hat x_N\}$, crosses the original sequence $\{x_1,s_1,\ldots,s_{N-1},x_{N}\}$ at most once and from above.
	\end{enumerate}
\end{proposition}

Part (i) of the proposition extends the single-dipped property to S-shaped value functions.  Intuitively, 
the single-dipped property under convex value functions means that the information designer puts less scrutiny on the two ends of the state space. When the value function is S-shaped, the last subinterval $(s_{N-1},1)$ covers the concave region and therefore should deserve even less scrutiny. Hence, the single-dipped property preserves.

Part (ii) of Proposition \ref{intensityorder-S} shows that the uniqueness result stated in Proposition \ref{unique} can be extended to S-shaped value functions as long as $u''$ is logconcave when positive and $f$ is logconcave given that $s_{N-1}<t$.  
Parts (iii) and (iv) mirror the comparative statics results in Propositions \ref{csforcutoffslr} and \ref{csforcutoffsuv}, respectively. Note that we allow $u''$ and $\hat{u}''$ to have different inflection points, $t$ and $\hat{t}$.  In Proposition \ref{intensityorder-S}(iii), $\hat u'' \succeq_{lr} u''$ implies $\hat t  \ge t$, meaning that $\hat u$ has a smaller concave region than $u$. In the extreme case, $\hat u$ can be convex (represented by some $\hat t>1$), and the statement remains valid. 
In Proposition \ref{intensityorder-S}(iv), $u'' \succeq_{uv} \hat u''$ implies either $t \le \hat t \le p$ or  $p \le \hat t \le t$. That is, if both inflection points are to the left of $p$, then $\hat u$ has a smaller concave region than $u$. Otherwise, if both inflection points are to the right of $p$, then $\hat u$ has a larger concave region than $u$.

\examplepart{b (continued)} 
In the purchase decision example, when the density $h$ of cost shock is single-peaked, the seller's value function $\hat u$ is S-shaped.
Proposition \ref{intensityorder-S}(iii) applies because $u'' \succeq_{lr} \hat u''$. Intuitively, a falling tail of the cost shock density $h$ implies that better information about higher states has diminishing returns, as the probability that the consumer will be on the margin between buying and not buying falls when mean quality is already high.  Compared to the previous case (Example \ref{pd}a), where $\hat u$ is convex and the density $h$ is rising, the seller has even less incentive than does the buyer to provide finer information on higher states.

\section{General Value Functions}
\label{generalsection}

The coarse information design problem with convex or S-shaped value functions is relatively tractable because interval-partitional information structures are optimal.  With general value functions,
it is possible that the optimal information structure entails information policies that send more than one signal within a partitional subinterval, and the analysis would become more involved.  In this section, we restrict our discussion to interval-partitional information structures, and provide a generalization of the dual expectations property established under convex and S-shaped value functions in previous sections.
Moreover, when the interval partition is indeed optimal among all possible information structures, we establish a general convexity property about the minorant function.

Consider interval-partitional information structures with $N$ intervals.  The minorant function associated with the optimal information structure within this set must be continuous. 
Yet it is not necessarily everywhere below the value function.

\begin{proposition} 
	\label{generaloptimal}
	For any general value function $u$, the optimal interval-partitional information structure with $N$ intervals, characterized by $\{ s_k \}_{k=0}^{N}$ and $\{ x_k \}_{k=1}^{N}$, must satisfy:
	\begin{itemize}
		\item[(i)]  $x_k = \phi(s_{k-1},s_{k})$, for $k=1,\ldots,N$; 
		\item[(ii)]  $u(x_{k})+u'(x_{k}) (s_k-x_{k})=u(x_{k+1})+u'(x_{k+1}) (s_k-x_{k+1})$, for $k=1,\ldots,N-1$.
	\end{itemize}
\end{proposition}

The second part of Proposition \ref{generaloptimal} establishes the continuity of the minorant function, which further generalizes the dual expectations property in the following sense.
When $u'(x_{k+1})\ne u'(x_{k})$, we can restate the continuity of minorant function 
as $s_k=\mu(x_{k},x_{k+1})$, where $\mu(\cdot,\cdot)$ is given by equation (\ref{barycenter}) that defines the the barycenter under the signed measure $u'$. 
Thus, the dual expectations property continues to hold as the barycenter is well defined.
On the other hand, if $u'(x_{k+1})=u'(x_{k})$,  then the continuity of the minorant function implies that  
\[
u'(x_{k})=\frac{u(x_{k+1})-u(x_{k})}{x_{k+1}-x_k}=u'(x_{k+1}).
\] 
In this case, the solution is characterized by a bi-tangent line, i.e., the minorant function is affine on $[s_{k-1},s_{k+1}]$.

For the next result, define for $x \in [x_1,x_N)$ the function,
\[
\overline{u}_G(x) :=
\sum_{k=1}^{N-1}\left[ u(x_k) + \frac{u(x_{k+1})-u(x_k)}{x_{k+1}-x_k} (x - x_k)\right] \mathbb{I}_{[x_k,x_{k+1})}(x),
\]
where $(x_1,\ldots,x_{N})$ are the induced posterior means under interval-partitional information structure $G$.  The graph of $\overline{u}_G$ is piecewise affine and connects the adjacent posterior means induced by $G$.  In the case that $u$ is convex, it is obvious that both $\overline{u}_G$ and $\underline{u}_G$ are convex and they satisfy:
\[
\overline{u}_G(x) \ge u(x) \ge \underline{u}_G(x).
\]
For general value functions $u$, the minorant function $\underline{u}_G$ need not be everywhere below $u$.  Nevertheless, the following properties hold.

\begin{proposition}\label{generalminorant}
	For any general value function $u$,  suppose an interval-partitional information structure $G$ with $N$ subintervals is optimal. Then 
	\begin{itemize}
		\item[(i)] $\overline{u}_G(x) \ge u(x)$ for all $x \in [x_1,x_N)$, and $\overline{u}_G$ is convex.

		\item[(ii)] The minorant function $\underline{u}_G$ is convex. Moreover, if the value function is convex or S-shaped, the minorant function has strictly increasing slopes for each affine segment, i.e., $u'(x_{k+1}) > u'(x_{k})$ for $k=1,\dots,N-1$. 
	\end{itemize}
\end{proposition} 

The dual expectations characterization for S-shaped value functions relies on 
the barycenter $\mu(\cdot)$ in equation (\ref{barycenter}) being well defined.  This is established further in Proposition \ref{generalminorant}(ii), because it ensures that the denominator in (\ref{barycenter}) is strictly positive.
Indeed, as long as the optimal discrete information structure $G$ is interval-partitional, 
the denominator in (\ref{barycenter}) is non-negative for general value functions.

\section{Discussion}
\label{Discussion}

\subsection{Convex piecewise affine value functions}\label{section:piece-wiseaffine} 

When the set of actions is finite, the value function becomes piecewise affine. In this case, the minorant function corresponding 
to an interval partition may not be unique because the subgradient at any kink of the value function is not unique. 
Nevertheless, we show that this concern is not an issue when the piecewise affine function is convex.  
An important implication is that 
the optimal interval partition will not generate a posterior mean that keeps the information designer indifferent between two adjacent actions (where a kink arises).

Suppose the value function $u$ has $M$ kinks. Then it can be represented by
\begin{equation*}
	u(x) = \sum_{m=0}^{M} \left[ u(\upsilon_m)+ u'(\upsilon_m)(x-\upsilon_m) \right] \mathbb{I}_{\left[\upsilon_{m},\upsilon_{m+1}  \right]} (x),
\end{equation*}
where $\{\upsilon_m\}_{m=0}^{M+1}$ includes all the kink points, and the end points $\upsilon_0=0$, $\upsilon_{M+1}=1$. 
Throughout this section, we use $u'(\cdot)$ to represent the right derivative whenever the left derivative is not equal to the right.
As before, we invoke the normalization that $u'(\upsilon_0)=0$ and $u'(\upsilon_M)=1$.
To avoid triviality, we assume that the number of signals $N$  satisfies $M \ge N \ge 2$. The next result shows that the posterior means of the optimal interval partition cannot fall into the set of kink points of $u$.

\begin{lemma}\label{kinks}
	Under the optimal interval partition $G^*$ for convex piecewise affine value functions, $
	\{x_i \}_{i=1}^N \cap \{ \upsilon_m\}_{m=1}^M = \emptyset.$
\end{lemma}

To understand the intuition, consider the case where the information designer is also the decision maker. 
He seeks to maximize the value of information. However, the value of information is minimized when the decision maker is indifferent between actions---which occurs at the kink points of $u$.  Therefore, if a subinterval yields a posterior mean that lies at a kink, the designer can always locally adjust the interval cutoffs to produce a different posterior mean, thereby obtaining a strictly higher information value.

With Lemma \ref{kinks},  we can uniquely pin down the minorant function in the same way as it is defined in equation (\ref{minorantu}),
where $u'(x_i)$ is just the derivative, 
since the left derivative coincides with the right derivative. 
Then we can apply exactly the same argument as in the smooth case to argue that, under the optimal interval partition, $\underline{u}_{G^\ast}(x)$ must be continuous.
Hence every $\{ s_i\}_{i=1}^{N-1}$ must be a meeting point of two affine segments of the value function. Let  $\mathcal{C}$ be the set of points where any two affine pieces meet with each other:
\[
\mathcal{C}= \left\{x \in [0,1]: u(\upsilon_i)+ u'(\upsilon_{i})(x-\upsilon_i)=u(\upsilon_j)+ u'(\upsilon_{j})(x-\upsilon_j) \text{ for } 
0\le i<j\le M\right\}.
\]

\begin{corollary}
	Under optimal information structure $G^\ast$ for convex piecewise affine value functions,   $\{ s_i\}_{i=1}^{N-1} \subseteq \mathcal{C}$.
\end{corollary}

Under the optimal information structure $G^\ast$, all interval cutoffs $\{ s_i\}_{i=1}^{N-1}$ are selected from the set $\mathcal{C}$. Hence, the number of information structures that 
we can search over to find the optimal one is finite. 

\subsection{Bimodal curvature and double-dipped interval widths}
\label{section:Beyond}

A logconcave environment ensures both $f$ and $u''$ are single-peaked, which in turn implies that the interval widths in the optimal information structure are single-dipped. Whether the value function curvature $u''$ is single-peaked or not depends on the nature of the underlying decision problem.  However, one of the basic insights of this paper---namely, the optimal information structure pays more scrutiny to the part of the state space where value function curvature is relatively high---does not depend on the logconcavity assumption. 

To illustrate, suppose a decision maker chooses to consume $a$ units of a good to maximize her 
expected utility $\mathrm{E}[B(a)-\theta a]$, 
where $B(a)$ is the intrinsic benefit from consumption, and $\theta$ represents per-unit cost. 
Her indirect value function will have curvature $u''(x)=-\mathrm{d}a^*(x)/\mathrm{d}x$, where $a^*(x)$ is the ``demand'' for this good when the expected cost is $x$.  The shape of the demand function $a^*(\cdot)$ generally depends on $B'^{-1}(\cdot)$, the inverse of the marginal benefit function.  Panel (a) of Figure \ref{numerical} shows a demand function that takes a logit form.  In this panel, demand is most sensitive to cost at intermediate levels of $x$.  If the density $f$ of the state is uniform, then the center of scrutiny is near 0.5.  Panel (b) shows a different possibility, where demand is more cost-sensitive for values of $x$ near 0 and 1.  The value function curvature $u''$ corresponding to this demand pattern is bimodal.  The corresponding optimal information structure exhibits interval widths which are double-dipped: the decision maker pays closest scrutiny to states near 0 and 1 because 
her action responds more to the state at the two ends of the state distribution. Finally, if the demand function is linear, then $u''$ will be constant.  With uniform prior density $f$, linear demand will lead to an optimal information structure in which all the intervals are evenly spaced.

\begin{figure}[t]
	\centering
	\includegraphics[width=10cm]{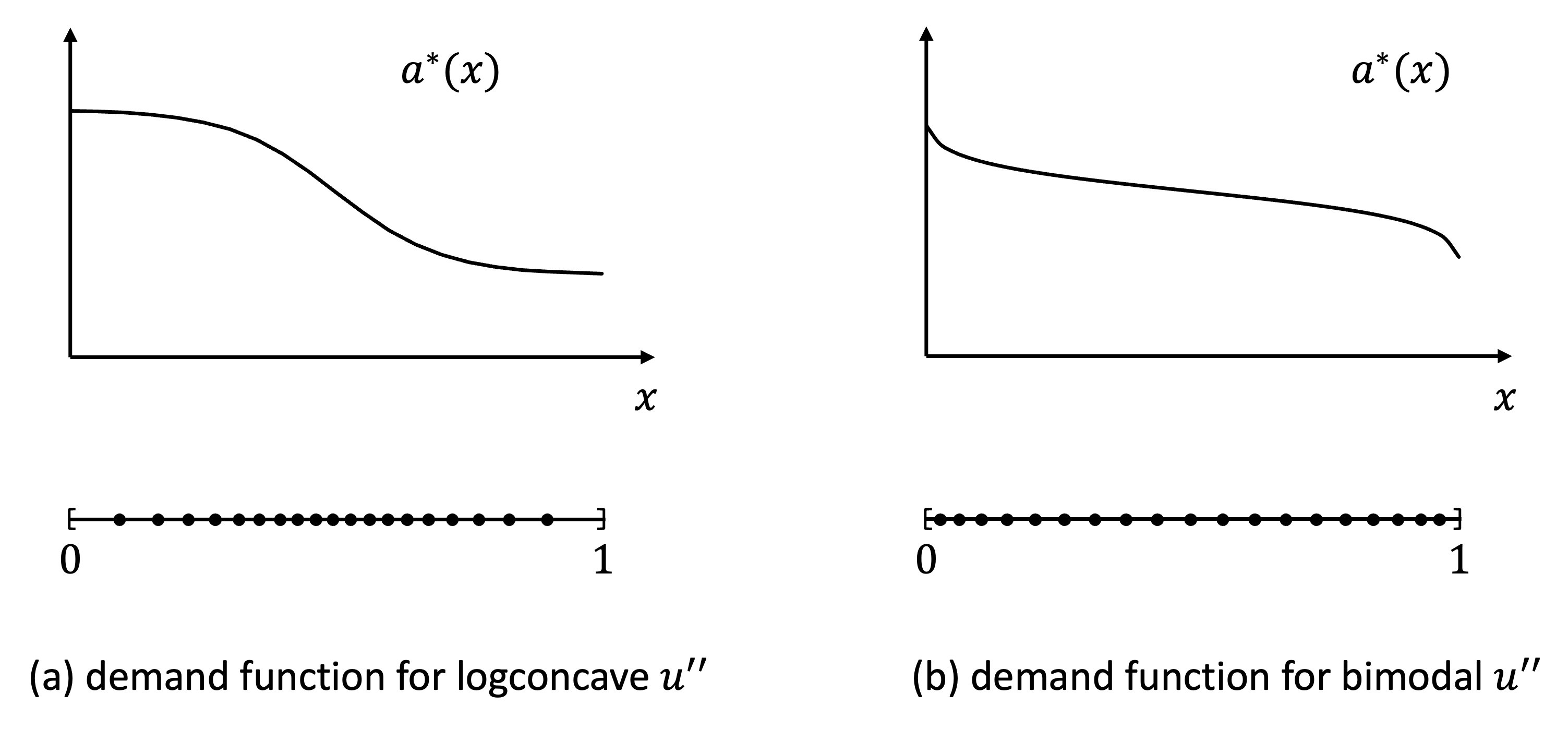}
	\caption{\small The demand function in panel (a) is $a^*(x)=\frac{1}{1+e^{10x+5}}+\frac{1}{2}$; it is most cost-sensitive near the middle part of the state distribution.  The demand function in panel (b) is $a^*(x)=\log\frac{0.99+0.98x}{0.01+0.98x} + 5$; it is most cost-sensitive near the two ends of the state distribution.  The line segments at the bottom of these two panels show the corresponding optimal partitioning of the state space.}
	\label{numerical}
\end{figure}

\subsection{Connection with cheap talk}
\label{section:Cheaptalk}

Our model of coarse information design bears some resemblance to the standard cheap talk model of \cite{CrawfordSobel1982}, where information coarsening arises endogenously due to incentive compatibility constraints. In \cite{CrawfordSobel1982}, 
to verify an interval-partitional information structure is an equilibrium strategy, we only need to verify that the cutoff types are indifferent between two adjacent actions taken by the receiver under two adjacent posterior means.  For example, suppose that the sender's payoff is 
$-(a-\kappa_1(\theta-\kappa_0))^2$ 
and the receiver's payoff is $-(a-\theta)^2$, 
where $a$ is the action taken by the receiver,
and $\kappa_1 \ge 1$, $\kappa_0 \in [0,1]$.
Then, the cheap-talk equilibrium, given a fixed number of signals $N$, is determined by the following system of equations:
\begin{align*}
	x_k &= \phi(s_{k-1},s_{k}) \qquad  \text{for } k=1,\ldots,N;   \\
	s_k &= \frac{x_k+x_{k+1}}{2\kappa_1}+\kappa_0
	\qquad \text{for } k=1,\ldots,N-1. 
\end{align*}
The second set of equations are the indifference conditions for the cutoff types.  They are different from the optimality conditions (\ref{mu}) for the coarse information design problem.
Nevertheless this system of equations is still a monotone mapping satisfying Lemma \ref{keylemma} if the prior density $f$ is logconcave. This means that some of our results should align with the equilibrium characterization in this cheap talk model, such as the existence of a center of scrutiny. 
The insight driving this result in the cheap talk model, 
however, differs from ours. 
Specifically, the sender's bias is smaller when $\theta$ is 
close to $\kappa_0$ than when it is farther away  \citep{Gordon2010,DeimenSzalay2019,DeimenSzalay2023}.
Hence, credibility is easier to obtain for states close to $\kappa_0$ and the intervals are narrower in that part of the state space.
In contrast, in our coarse information design problem, the center of scrutiny is mainly determined by the value function curvature and the prior density of the state. 
With quadratic payoffs, the value function curvature $u''(x)$ of the sender is a constant that does not depend on $x$.
Hence, the values of the parameters $\kappa_1$ and $\kappa_0$ have no effect on the optimal information structure.
Indeed, in this example, the sender-optimal coarse experiment with $N$ signals is identical to the receiver-optimal coarse experiment with the same number of signals, because the value function curvature of the receiver is also independent of $x$.

Recent developments in cheap-talk games have also explored the optimal information design by the sender. This literature assumes that an uninformed sender commits to an experiment, and sends an unverifiable message about the experimental outcome.  
When the loss function is quadratic with 
$\kappa_1=1$, the sender's problem can be transformed into a problem similar to ours,  with an additional set of truth-telling constraints that require the distance between every two adjacent messages to be at least $2\kappa_0$, and with the number of messages $N$ being endogenously determined through this constraint. 
As elaborated in \cite{KreutzkampLou2024},
this additional constraint makes bi-pooling a possible candidate solution despite the quadratic value function. 
In our model, $N$ is exogenously fixed and there are no truth-telling constraints. Hence, even with the same payoff structures, the optimal experiment in our setup has an interval-partitional structure.

\section{Conclusion}

The literature on Bayesian persuasion often assumes no restrictions on the set of experiments that can be chosen, making the information design problem trivial when the value function is convex. One approach that has been taken to relax this assumption is to introduce a posterior-based cost of acquiring information that depends on the experiment chosen 
\citep{CaplinDean2015,Matejka2015,BloedelSegal2021, RavidRoeslerSzentes2022, MatyskovaMontes2023};
see \cite{Denti2022} and \cite{DentiMarinacciRustichini2022}
for a critical assessment of that approach.  In this paper, we adopt an alternative track by imposing a finiteness constraint on the signal space to reflect the coarseness of information structures.  We take information coarsening as given and consider the best way of doing it.  This approach naturally leads to a research agenda that investigates which parts of the state space deserve most attention or scrutiny in information acquisition decisions.	Our paper only takes a first stab at this research agenda by considering a simple environment with one-dimensional state where the belief about the state matters only through the mean.  Many questions regarding sequential information acquisition,  competitive and complementary constrained information provision, or coarse information design under more varied payoff or informational environments remain to be explored.

\newpage 
\appendix

\section*{Appendix}

\begin{proof}[{\bf Proof of Lemma \ref{keylemma}}]

	Take any $\eps \ge 0$.  We can write
	\begin{align*}
		\phi(a+\eps,b+\eps) &= \int_{a+\eps}^{b+\eps} x \frac{f(x)}{F(b+\eps)-F(a+\eps)}\,\rmd x = \int_a^b (x' + \eps) \frac{f(x'+\eps)}{F(b+\eps)-F(a+\eps)}\,\rmd x', \\
		\phi(a,b)+\eps  &= \int_a^b (x'+\eps) \frac{f(x')}{F(b)-F(a)} \,\rmd x'.
	\end{align*}
	The difference is
	\begin{equation*}
		\int_a^b x\left(\frac{f(x+\eps)}{F(b+\eps)-F(a+\eps)}-\frac{f(x)}{F(b)-F(a)}\right)\, \rmd x.
	\end{equation*}
	Due to the logconcavity of $f$, the first density is dominated by the second in the likelihood ratio order. Hence, the difference is nonpositive.  Since $\phi(\cdot,\cdot)$ is strictly increasing in each argument, 
	\[
	\phi(a+\eps_1,b+\eps_2) \le \phi(a+\max\{\eps_1,\eps_2\},b+\max\{\eps_1,\eps_2\}) \le \phi(a,b) + \max\{\eps_1,\eps_2\},
	\]
with strict inequality if $\eps_1 \ne \eps_2$.
\end{proof}

\begin{proof}[{\bf Proof of Proposition
		\ref{convexinterval}}]
	The objective function (\ref{auxiliaryprogram}) is the $u''$-weighted area below $I_G$. Since $u''$ is nonnegative, an information structure $G'$ performs better than another information structure $G$ if $I_{G'}\ge I_G$. Consider an arbitrary induced distribution $I_G$ with	a segment strictly below $I_F$ that does not contain a tangency point. We can construct an alternative distribution $I_{G^{\prime}}$ by pointwise increasing this segment until it crosses the adjacent linear pieces (see the left panel of Figure \ref{figure:convexinterval}). By construction, $I_{G^{\prime}}$ lies above $I_G$ and remains below $I_F$. Moreover, the number of induced posterior means remains unchanged. Thus, the information structure associated with $I_{G^{\prime}}$ is feasible and yields a strictly higher value of the objective function. It follows that under the optimal information structure, every segment of $I_G$ must be tangent to $I_F$. By Lemma \ref{lemma:conditional_expectation}, such an $I_G$ can be implemented via an interval-partitional information structure.
	
	\begin{figure}[t]
		\centering
		\includegraphics[width=14cm]{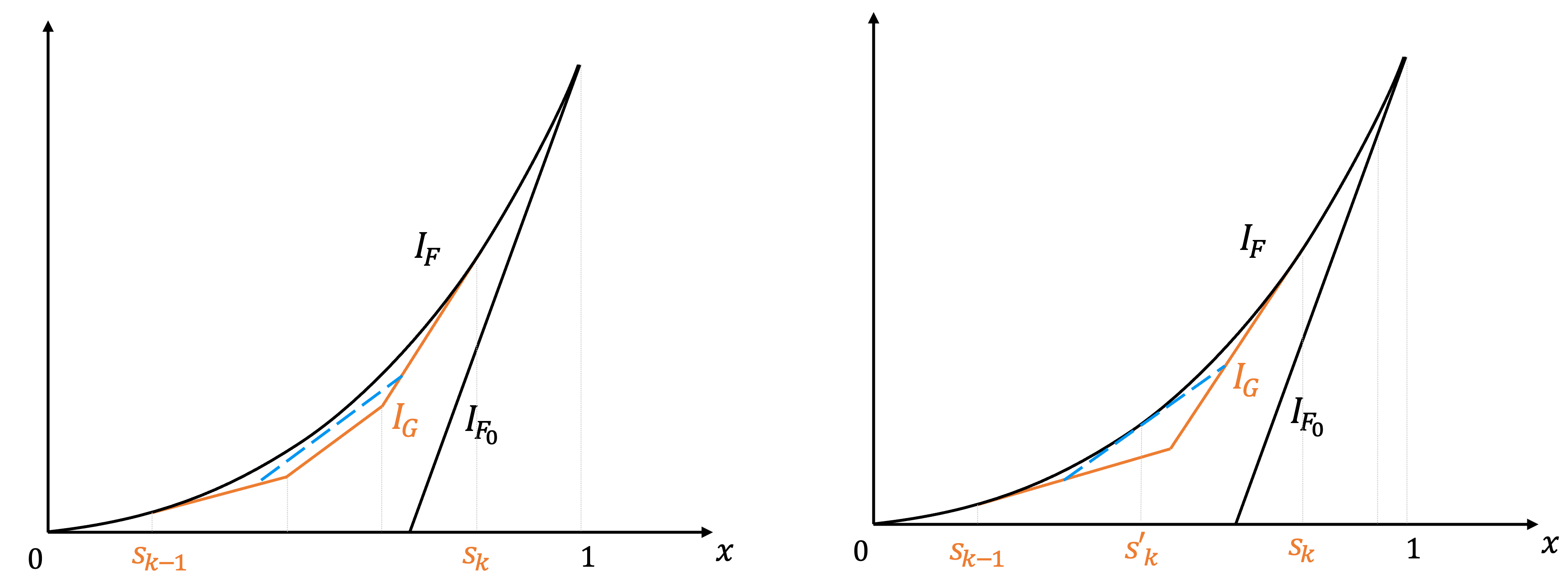}
		\caption{Illustration of Proposition 
			\ref{convexinterval}
		}
		\label{figure:convexinterval}
	\end{figure}

	Next, suppose the optimal interval partition induces $N'<N$ posteriors.	Then the sender can divide some interval $[s_{k-1},s_k)$ into two pieces $[s_{k-1},s'_k)\cup [s'_k,s_k)$ (see the right panel of Figure \ref{figure:convexinterval}).  This new integral distribution with an additional kink point is everywhere higher than the original integral distribution and will produce a higher value of the objective function when $u$ is convex.
\end{proof}

\begin{proof}[{\bf Proof of Proposition \ref{peak}}]
	
	We prove (iii) first.  The functions $\phi(\cdot)$ and $\mu(\cdot)$ are conditional expectations.  Therefore, $\phi(a,b) \ge (a+b)/2$ and $\mu(a,b) \ge (a+b)/2$ for any $a < b \le \min\{m_f,m_u\}$, because both $f$ and $u''$ are increasing in the relevant region.  An induction argument establishes that $w_1 \ge w_2 \ge \ldots \ge w_{k-1}$.  Similarly,  $\phi(a,b) \le (a+b)/2$ and $\mu(a,b) \le (a+b)/2$ for any $\max\{m_f,m_u\} \le a < b$, because both $f$ and $u''$ are decreasing in the relevant region.  An induction argument establishes that $w_N \ge w_{N-1} \ge w_{j+1}$.  Together, they imply that there exists an $i^* \in \{k-1,\ldots,j+1\}$ such that the sequence of widths attains a minimum at $w_{i^*}$.
	
	A constant $u''$ is a special case of single-peaked $u''$.  We can arbitrarily set the mode of $u''$ at $m_f$, and part (i) follows.  Similarly, if $f$ is a constant, we can arbitrarily set its mode at $m_u$, and part (ii) follows.
\end{proof}

\begin{proof}[{\bf Proof of Proposition \ref{generaloptimal}}]
	Part (i) follows from the definition of an interval-partitional information structure.  Under an interval-partitional information structure $G$,
	\[
	I_G(x)=\begin{cases}
		0 &\text{if } x \in [0,x_1)\\
		I_F(s_k)+F(s_k)(x-x_k) &\text{if } x \in [x_k,x_{k+1}) \text{ for } k=1,\ldots,N-1\\
		I_F(1)  + (x-1) &\text{if } x \in [x_{N},1].
	\end{cases}
	\]
Therefore, the objective function (\ref{auxiliaryprogram}) can be written as a function of the interval cutoffs, 
	\[
	\sum_{k=1}^{N-1} \int_{x_k}^{x_{k+1}} u'' (x)\left(I_F(s_k) +F(s_k)(x-s_k)\right) \, \rmd x + \int_{x_N}^{1} u'' (x)\left(I_F(1) + (x-1)\right) \, \rmd x,  
	\]
where $x_k=\phi(s_{k-1},s_k)$. The first-order condition  with respect to each interval cutoff $s_k$ leads to,
	\[
	\int_{x_{k}}^{x_{k+1}} u''(x) (x-s_k)\, \rmd x=0, \quad k=1,\ldots,N-1. 
	\]
Upon integration-by-parts, the above condition leads to the continuity of the minorant function:
	\[
	u(x_{k})+u'(x_{k}) (s_k-x_{k})=u(x_{k+1})+u'(x_{k+1}) (s_k-x_{k+1}), \quad k=1,
	\ldots,N-1. \qedhere
	\]
\end{proof}

\begin{proof}[{\bf Proof of Proposition \ref{generalminorant}}]
	We first prove that $\overline{u}_G(x) \ge u(x)$ in (i). Let 
	\[
	D(a,b):=\frac{u(b)-u(a)}{b-a}.
	\]
Suppose to the contrary that under the optimal interval partition $G$, there exists $\hat x \in (x_k,x_{k+1})$ such that $u(\hat x) > u(x_k)+D(x_k,x_{k+1})(\hat x-x_k)$.  This would imply $D(x_k,\hat x) > D(x_k,x_{k+1})$.  We then construct another information structure that is strictly better than $G$, which would contradict its optimality.
	
	\textit{Case 1:} $\hat x > \phi(s_{k-1},s_{k+1})$. Consider an information structure $\hat G$ such that in the interval $(s_{k-1},s_{k+1})$, we let the experiment induce a pair of posterior means, $x_k$ or $\hat x$. Because $s_{k-1}<x_k < \phi(s_{k-1},s_{k+1}) < \hat x<x_{k+1}<s_{k+1}$, by Lemma 1 of \cite{ArieliEtAl2023}, such an experiment constitute a bi-pooling policy and is feasible. The information structure $\hat G$ is otherwise identical to $G$ for states outside $(s_{k-1},s_{k+1})$.
Observe that 
	\begin{align*}
		\int_{s_{k-1}}^{s_{k+1}} u(x)\,\rmd G &= \left(F(s_{k+1})-F(s_{k-1})\right) \left(u(x_k)+D(x_k,x_{k+1})(\phi (s_{k-1},s_{k+1}) - x_{k-1})\right)\\
		&< \left(F(s_{k+1})-F(s_{k-1})\right) \left (u(x_k)+D(x_k,\hat x)(\phi (s_{k-1},s_{k+1})- x_{k-1})\right)
		=\int_{s_{k-1}}^{s_{k+1}} u(x)\,\rmd \hat G.
	\end{align*}
Thus $\hat G$ is strictly better than $G$.
	
	\textit{Case 2:}
$\hat x < \phi(s_{k-1},s_{k+1})$.
This implies $D(\hat x,x_{k+1}) < D(x_k,x_{k+1})$.  Consider an information structure $\hat G$ such that in the interval $(s_{k-1},s_{k+1})$, the experiment will induce a pair of posterior means, $\hat x$ and $x_{k+1}$. Because $x_k<\hat{x}<\phi(s_{k-1},s_{k+1})<x_{k+1}$, this new information structure $\hat G$ is also feasible. With a similar argument as above,  $\hat G$ gives a strictly higher payoff than $G$. 
	
	\textit{Case 3:}
$\hat{x}=\phi(s_{k-1},s_{k+1})$. Because $u$ is continuous, we can always find a nearby $x'$ such that $x'<\phi(s_{k-1},s_{k+1})$ and $u(x')>u(x_k)+D(x_{k},x_{k+1}) (x'-x_k)$. Therefore, we can adopt the same argument as above to construct a better information structure.

	To prove part (ii), note that $ u(x_{k})+u'(x_{k}) (s_k-x_{k})=u(x_{k+1})+u'(x_{k+1}) (s_k-x_{k+1})$ for $k=1,\ldots,N-1$   can be re-arranged into:
	\begin{equation*}
		D(x_k,x_{k+1})=\frac{s_k-x_k}{x_{k+1}-x_k} u'(x_k) + \frac{x_{k+1}-s_k}{x_{k+1}-x_k}u'(x_{k+1}).
	\end{equation*}
Thus, $D(x_k,x_{k+1})$ must lie between $u'(x_k)$ and $u'(x_{k+1})$. Suppose $u'(x_k) > u'(x_{k+1})$, then we would have $u'(x_k) > D(x_k,x_{k+1}) > u'(x_{k+1})$.  But this would imply $u(x) > u(x_k)+D(x_k,x_{k+1})(x-x_k)$ for $x$ sufficiently close to (and larger than) $x_k$, which contradicts part (i). This shows that $u'(x_k) \le u'(x_{k+1})$ for all $k=1,\ldots, N-1$ under the optimal interval partition. 
	
	When the value function is convex, it is obvious that $u'(x_k) < u'(x_{k+1})$ holds. Now we argue that the same strict inequality holds when the value function is S-shaped. Suppose $u'(x_k) = u'(x_{k+1})$, then $u'(x_k) =D(x_k,x_{k+1})=u'(x_{k+1})$, which further implies that $u(x_{k+1})=u(x_k)+u'(x_k) (x_{k+1}-x_k)$ (bi-tangency). Note that $u''$ is single-crossing from above; therefore $u'(x)> u'(x_k)$ for all $x\in(x_k,x_{k+1})$. This implies that
	\[
	u(x_{k+1})=u(x_k)+\int_{x_k}^{x_{k+1}} u'(x)\, \rmd x>u(x_k)+\int_{x_k}^{x_{k+1}} u'(x_k)\, \rmd x=u(x_k)+u'(x_k) (x_{k+1}-x_k).
	\] 
This contradicts the bi-tangency condition, and we have established part (ii).
	
	Finally, the proof of part (ii) shows that $D(x_k,x_{k+1})$ must lie between $u'(x_k)$ and $u'(x_{k+1})$.  Since we have established that $u'(x_k)$ is nondecreasing in $k$, $D(x_k,x_{k+1})$ must also be nondecreasing in $k$, and therefore $\overline{u}_G$ is convex.  This completes the proof of part (i).
\end{proof}

\begin{proof}[{\bf Proof of Lemma \ref{kinks}}]
	Suppose to the contrary that $x_i=\upsilon_m$ for some $i<N$. (The case where $i=N$ can be proved in the same way).  Let $\underline{u}_{i}$ be the subgradient line at $x_{i}$, with the slope equal to the \emph{left derivative} at that point, i.e., let $\underline{u}_{i}(x) :=u(x_{i}) + u'(\upsilon_{m-1}) (x- x_{i})= u(\upsilon_m) + u'(\upsilon_{m-1}) (x- \upsilon_m)$. Now consider the subgradient line $\underline{u}_{i+1}$ at $x_{ i+1}$, with slope equal to the \emph{right derivative} at that point, i.e., $\underline{u}_{i+1}(x):= u(x_{ i+1}) + u'(x_{ i+1})(x- x_{ i+1})$. Note that $x_{i+1}\ge\upsilon_{m+1}$ because otherwise either $x_i$ or $x_{i+1}$ is redundant.  Let	$\{s_i\}_{i=1}^{N-1}$ be the 	interval cutoff points under the optimal information structure $G^*$. There are three cases to consider.

	\textit{Case 1:} $\underline{u}_{ i}(s_{ i}) \ge  \underline{u}_{i+1}(s_i)$. 
Then there exists some $\varepsilon>0$ such that $s_i+\varepsilon< \upsilon_{m+1}  $ and $\underline{u}_{i}(x) > \underline{u}_{i+1}(x)$ for all $x \in [s_i, s_i+\varepsilon]$. Consider another information structure $\hat G$, which is identical to $G^*$ except that the cutoff point $s_i$ is replaced by $s_i+\varepsilon$.  Thus, $\hat x_{i}=\phi(s_{i-1}, s_{i}+\varepsilon) $ and $\hat x_{i+1}=\phi(s_{i}+\varepsilon, s_{i+1}) $. Note that we can always choose $\eps$ such that $\hat{x}_{i+1}$ remains in the same affine segment as $x_{i+1}$. We want to show that $\hat G$ gives a higher payoff than $G$. To see this, 
	\begin{align*}
		\int_{s_{i-1}}^{s_{i+1}} u(x) \, \rmd {\hat G}(x) &=\int_{s_{i-1}}^{s_{i}+\varepsilon}  \left[ u(\hat x_{i}) +  u'(\upsilon_{m-1}) (x- \hat x_{i}) \right] \, \rmd {\hat G}(x) +  \int_{s_i+\varepsilon}^{s_{i+1}} \underline u_{i+1}(x)\, \rmd {\hat G}(x)\\
		&=\int_{s_{i-1}}^{s_{i}+\varepsilon}  \left[ u(\hat x_{i}) +  u'(\upsilon_{m-1}) (x- \hat x_{i}) \right] \, \rmd {F}(x) +  \int_{s_i+\varepsilon}^{s_{i+1}} \underline u_{i+1}(x)\, \rmd {F}(x)\\
		&>\int_{s_{i-1}}^{s_{i}}  \underline{u}_{i}(x) \, \rmd {F}(x) +  \int_{s_i}^{s_{i+1}} \underline u_{i+1}(x)\, \rmd {F}(x)
		=\int_{s_{i-1}}^{s_{i+1}} u(x) \, \rmd {G^*}(x).
	\end{align*} 
The inequality in the above comes from $\hat{x}_i>x_i$ and from the fact that $u(x)$ is piecewise convex and affine; thus the graph of the two segments $u(\hat x_{i}) +  u'(\upsilon_{m-1}) (x- \hat x_{i})$ on $[s_{i-1}, s_{i}+\varepsilon ]$ and $\underline{u}_{i+1}(x)$ on $[s_{i}+\varepsilon,s_{i+1}  ]$ are above the two segments $\underline{u}_{i}(x)$ on $[s_{i-1}, s_{i} ]$ and $\underline{u}_{i+1}(x)$ on $[s_{i},s_{i+1}  ]$.  Since $G^\ast$ and $\hat G$ yields the same payoff for all states outside $[{s_{i-1}},{s_{i+1}}]$,  $G^\ast$ is strictly suboptimal, which is a contradiction. 
	
	\textit{Case 2:} $\underline{u}_{ i}(s_i) < \underline{u}_{i+1}(s_{ i})$. 
There exists some $\varepsilon>0$ such that $s_{ i}-\varepsilon> x_{  i}$ and $\underline{u}_{  i}(x) < \underline{u}_{{ i}+1}(x)$ for all $x \in [s_{ i}-\varepsilon, s_{ i}]$. Consider another information structure $\hat G$, which is identical to $G^*$ except that $s_i$ is replaced by $s_i-\varepsilon$.  With a similar argument as in case 1, we can show that $\hat G$ gives a higher payoff than $G^*$. 
\end{proof}

\newpage
\bibliographystyle{aer}
\bibliography{coarse}

\newpage 

\newpage
\thispagestyle{empty}
\vspace*{2in}

\centerline{[This page is intentionally left blank.]}

\newpage
\setcounter{page}{1}
\renewcommand{\thepage}{\roman{page}}

\section*{Online Appendix: Proof of results for S-shaped value functions}\label{appendix:s-shaped}

We first establish a claim regarding a property of the optimal interval-partitional information structure under S-shaped value functions.

\begin{claim}\label{sN-1<inflection-S}
Under the optimal $N$-interval partition $ G$, $s_{N-1}<t$.
\end{claim}
\begin{proof}
Suppose $s_{N-1} \ge t$. It implies $x_N > t$, and thereby $x_{N-1} < t$, for otherwise the line segment connecting $u(x_{N-1})$ and $u(x_N)$ is not above $u(x)$  contradicting Lemma \ref{generalminorant}(i). By concavity of $u$ on $[t,x_N]$ and the fact that $u(x_N)=\underline{u}_G(x_N)$, we have $u(t) < u(x_N) + u'(x_N)(t- x_N) \le \underline{u}_G(t) $,
where the last inequality is based on the convexity of the minorant function from Lemma \ref{generalminorant}(ii). On the other hand, by convexity of $u$ on $[x_{N-1}, t]$ and the fact that $u(x_{N-1})=\underline{u}_G(x_{N-1})$, we have $u(t) > \underline{u}_G(t)$. This contradicts the continuity of $\underline{u}_G(\cdot)$.  
\end{proof}

The next claim completes the proof of the last part of Proposition \ref{dual-S}.

\begin{claim}\label{sN-1lesstar} 
Under the optimal $N$-interval partition $G$, $s_{N-1} < s^*$.
\end{claim}

\begin{proof}
Let $x^\ast= \phi(s^\ast,1)$. By Theorem 1 in \cite{DworczakMartini2019}, $s^\ast<t<x^\ast$. Let $\underline{u}^\ast(x)=u(x^\ast)+u'(x^\ast)(x-x^\ast)$. \cite{DworczakMartini2019} show that $\underline{u}^\ast(x) \ge u(x)$ on $[s^\ast,1]$, with equality at $s^\ast$ and $x^\ast$.

Now suppose $s_{N-1}\ge s^\ast$  and therefore $x_N\ge x^\ast$. Consequently, $x_N \ge  x^\ast>t$. The concavity of $u$ on $(t,1]$ and the fact that $u(x_N)=\underline{u}_G (x_N)$ imply that 
$\underline{u}_G(x^*) \ge u(x^\ast)=\underline{u}^\ast(x^\ast)$ 
and $\underline{u}_G'(x)=u'(x_N)<u'(x^\ast)={\underline{u}^{\ast}}{'}(x)$ for all $x \in (s_{N-1},1]$. Therefore,
\[
\underline{u}_G(s_{N-1}^+) 
\ge \underline{u}^\ast (s_{N-1}). 
\]
On the other hand, by Claim \ref{sN-1<inflection-S}, $s_{N-1} < t$. 
Therefore,
\[
\underline{u}^\ast(s_{N-1}) 
\ge u(s_{N-1}) > \underline{u}_G(s_{N-1}^-) =u(x_{N-1})+u'(x_{N-1})(s_{N-1}-x_{N-1}),
\]
where the last inequality comes from the fact $u$ is convex on $[0,s_{N-1}]$ and $u(x_{N-1})=\underline{u}_G(x_{N-1})$. 
The two displayed inequalities contradict the continuity of the minorant function $\underline{u}_G$ required by Proposition \ref{generalminorant}(ii).
\end{proof}

\begin{proof}[{\bf Proof of Lemma \ref{keylemma-S}}]
We show that
$\frac{\partial \mu(a,b)}{\partial a}\le 1$ for all $a \in (0,t)$ such that $u'(b)-u'(a)>0$. 

Note that 
\[
\frac{\partial \mu(a,b)}{\partial a}=\frac{u''(a)\int_a^b  \left[ u'(b)-u'(x) \right]\, \rmd x}{(u'(b)-u'(a))^2}.
\]
Let $c \in (0, t)$ be the point such that $u'(c)=u'(b)$. 
Since $u'$ is quasiconcave, we have $u'(b)-u'(x) < 0$ for all $x \in (c,b)$.  Therefore,
\begin{align*}
	\frac{\partial \mu}{\partial a} &
	\le \frac{u''(a)\int_a^c  \left[ u'(b)-u'(x) \right]\, \rmd x}{(u'(b)-u'(a))^2} \\
	&\le \frac{1}{u'(b)-u'(a)} \int_a^c \frac{ u''(x)}{ u'(b)- u'(x)}\left( u'(b)- u'(x)\right)\, \rmd x 
	\\
	&= \frac{1}{u'(b)-u'(a)}\left(u'(c)- u'(a)\right)
	=1. 
\end{align*}
The second inequality above comes from 
logconcanvity of $u''$ on the convex region (and hence $\frac{u''(a)}{u'(b)-u'(a)}$ is increasing in $a$),
and from the fact that 
$u'(b)-u'(x)>0$ for all $x\in [a,c]$. 
\end{proof}

\begin{proof}[{\bf Proof of Proposition \ref{intensityorder-S}(i)}]
The proof of Proposition \ref{intensityorder} shows that $\Delta_i \geq \Delta_{i-1}$ implies $\Delta_{i+1} \geq \Delta_i$ within the convex region, namely, for all $i=2, \ldots, 2 N-4$. For S-shaped value function, it remains to be shown that $\Delta_i \geq \Delta_{i-1}$ implies $\Delta_{i+1} \geq \Delta_i$ for $i=2N-3$ and $i= 2N-2$. 

For $i=2N-3$, this is equivalent to showing that $w_{N-1}\ge d_{N-2}$ implies $d_{N-1}\ge w_{N-1}$. There are two cases: (i) $x_N> t$ and (ii) $x_N\le t$. 

Under case (i), it suffices to prove that $w_{N-1}\ge d_{N-2}$ implies $t-x_{N-1}\ge  w_{N-1}$, because it will further implies $d_{N-1}=x_N-x_{N-1}> t- x_{N-1}\ge w_{N-1}$. Suppose that $w_{N-1} \ge d_{N-2}$ but $t-x_{N-1} < w_{N-1}$. Then, 
\begin{align*}
	s_{N-1}=\mu(x_{N-1},x_N)&=\mu(x_{N-2}+d_{N-2},x_{N-1}+d_{N-1})\\
	&< \mu(x_{N-2}+d_{N-2},x_{N-1}+ t-x_{N-1}) \\
	&\le \mu(x_{N-2}+w_{N-1},x_{N-1}+t-x_{N-1}) \\
	&\le s_{N-2}+\max\{w_{N-1},t-x_{N-1}\}
	=s_{N-2}+w_{N-1} =s_{N-1},
\end{align*}
which is a contradiction.

Under case (ii), the proof is the same as that of Proposition \ref{intensityorder}. 

For $i=2N-2$, we want to show that $d_{N-1}\ge w_{N-1}$ implies $w_{N}\ge d_{N-1}$. Suppose we have $w_{N}< d_{N-1}$ instead. Then, 
\begin{align*}
	x_{N}=\phi\left(s_{N-1}, s_{N}\right) &=\phi\left(s_{N-2}+w_{N-1}, s_{N-1}+w_{N}\right)\\
	&<\phi\left(s_{N-2}+d_{N-1}, s_{N-1}+d_{N-1}\right) \leq x_N+d_{N-1}=x_{N},
\end{align*}   
a contradiction.
\end{proof}

\begin{proof}[{\bf Proof of Proposition \ref{intensityorder-S}(ii)}]
To show the uniqueness of solution to the equation system,
suppose there are multiple optimal interval-partitional information structures. By Claim \ref{sN-1<inflection-S}, the last cutoff $s_{N-1}$ in all solutions must fall strictly below $t$. 
Let $\{s'_k\}_{k=0}^N$ and $\{x'_k\}_{k=1}^N$ be the solution with the smallest value of $s'_{N-1}$ (if there are more than one such solution, pick an arbitrary one).  Consider another solution $
\{s''_k\}_{k=0}^N$ and $\{x''_k\}_{k=1}^N$ with $t > s''_{N-1} \ge s'_{N-1}$. Now remove the last two equations, $x_N=\phi(s_{N-1},1)$ and $s_{N-1}=\mu(x_{N-1},x_N)$ from the original equation system and consider the remaining equations:
\begin{align*}
x_k=\phi(s_{k-1},s_k) &\quad \text{for } k=1,\ldots,N-1\\
s_k=\mu(x_k,x_{k+1}) &\quad \text{for } k=1,\ldots,N-2,
\end{align*}
with $s_0=0$ and $s_{N-1}=s'_{N-1}$. By the same argument in Proposition \ref{unique}, $ \{ x'_k\}_{k=1}^{N-1}$ and $\{s'_k\}_{k=1}^{N-2}$ must be the unique solution to the reduced equation system. Similarly, $ \{ x''_k\}_{k=1}^{N-1}$ and $\{s''_k\}_{k=1}^{N-2}$ must be the unique solution to the reduced equation system with $s_0=0$ and $s_{N-1}=s''_{N-1}$. Suppose $s''_{N-1}=s'_{N-1}$.  This would imply that $(x'_1,\ldots,x'_{N-1})=(x''_1,\ldots,x''_{N-1})$.  Moreover, $x'_N=\phi(s'_{N-1},1)=\phi(s''_{N-1},1)=x''_N$.  Thus these two solutions would be identical, a contradiction.

Hence, it must be the case that $s''_{N-1} > s'_{N-1}$ and hence $x''_N > x'_N$. 
Corollary 2.5.2 in \cite{Topkis1998}
implies that $(x''_1,\ldots,x''_{N-1})$ is (element-wise) larger than $(x'_1,\ldots,x'_{N-1})$.  Moreover, we must have $s''_{N-1}-s'_{N-1} > x''_{N-1}-x'_{N-1}$.  To see this, suppose instead $s''_{N-1}-s'_{N-1} \le x''_{N-1}-x'_{N-1}$.  Then Lemma \ref{keylemma} would suggest $s''_{N-2} - s'_{N-2} \ge  x''_{N-1}-x'_{N-1}$.  Iterating the argument back leads to $x''_1 - x'_1 \ge s''_1-s'_1$, which contradicts Lemma \ref{keylemma}. Therefore, we can restate the inequality,
\[
s''_{N-1} - s'_{N-1} = \mu(x''_{N-1},x''_N)  - \mu(x'_{N-1},x'_N) > x''_{N-1}-x'_{N-1} .
\]
There are two possibilities.
Case (a): $x''_{N} - x'_N \le s''_{N-1}-s'_N$.  If $ t \ge x''_N > x'_N$ , then the inequality above already contradicts Lemma \ref{keylemma}. If $x''_N \ge t \ge x'_N$, then
\begin{align*}
 \mu(x''_{N-1},x''_N)  - \mu(x'_{N-1},x'_N) < \mu(x''_{N-1},t) - \mu(x'_{N-1},x'_N) &\le \max\{ x''_{N-1}-x'_{N-1}, t-x'_N\}\\
 & \le t -x'_N \le x''_N-x'_N,
\end{align*}
which contradicts the premise of Case (a).  If $x''_N > x'_N \ge t$, then 
\[
 \mu(x''_{N-1},x''_N)  - \mu(x'_{N-1},x'_N) \le \mu(x''_{N-1},x'_N) - \mu(x'_{N-1},x'_N) \le x''_{N-1}-x'_{N-1},
\]
where the last inequality comes from Lemma \ref{keylemma-S}. Yet this contradicts the original inequality again.  Thus the only remaining possibility is Case (b): $x''_{N-1}-x'_N > s''_{N-1}-s'_N$, 
but this inequality violates Lemma \ref{keylemma} because $f$ is  logconcave.
\end{proof}

Before proving Proposition \ref{intensityorder-S}(iii),
we make a few observations on the properties regarding the generalized likelihood ratio order. Define the barycenter for $u'$ and $\hat u'$ by $\mu(a,b)  :=\frac{\int_a^b  \theta u''(\theta) \, \rmd \theta}{\int_a^b  u''(\theta) \, \rmd \theta}$ and $\hat \mu(a,b)  := \frac{\int_a^b  \theta \hat u''(\theta) \, \rmd \theta}{\int_a^b  \hat u''(\theta) \, \rmd \theta}$, whenever the denominators are nonzero.

\begin{claim}\label{monotonicitymu-S}
For any $0<a<t\le b\le 1$ such that $u'(b)-u'(a)>0$, $\frac{\partial \mu (a,b)}{\partial b}<0$.
\end{claim}

\begin{proof}  
The derivative 
of $\mu$ with respect to $b$ 
is
\[
\frac{\partial \mu(a,b)}{\partial b} =\frac{u''(b) \int_a^b [u'(x)-u'(a)] \, \rmd x }{[u'(b)-u'(a)]^2}.
\]
This expression is negative because $u''(b)<0$; moreover $u'(b) > u'(a)$ implies $u'(x) > u'(a)$ for all $x \in [a,b]$.
\end{proof}

\begin{claim}\label{propertylr-S}
Suppose $\hat u'' \succeq_{lr} u''$. For any $a \in [0,t)$ and $b>a$, 
\begin{enumerate}[(i)]
	\item $u'(b)-u'(a) > 0$ implies that $\hat u'(b)-\hat u'(a)>0$;
	\item $u'(b)-u'(a) > 0$  implies $\hat \mu(a,b) \ge \mu(a,b)$.
\end{enumerate}
\end{claim}

\begin{proof}
The statements are trivially true if $\hat u$ is convex. Hence, we focus on the case where $\hat u$ is S-shaped. To show part (i), 
$\hat u'' \succeq_{lr} u''$ implies $\frac{\hat u''(s)}{\hat u''(a)} \ge \frac{ u''(s)}{ u''(a)}$ for all $s\ge a$, and hence
\[
\frac{\hat u'(b)-\hat u'(a)}{\hat u''(a)} = \int_a^b \frac{\hat u''(s)}{\hat u''(a)} \, \rmd s\ge \int_a^b \frac{ u''(s)}{ u''(a)} \, \rmd s =\frac{ u'(b)- u'(a)}{u''(a)}.
\]
Since $a <t \le \hat{t}$, we have $u''(a)>0$ and $\hat u''(a)>0$. Therefore, $u'(b)-u'(a) >0$ implies $\hat u'(b)-\hat u'(a) >0$.

Now we show part (ii).  If $b \le t$, then the inequality follows directly from the property of standard likelihood ratio order for positive density functions. 
If $ b \in (t,\hat t]$, then $\mu(a,b) \le \mu (a,t)\le \hat \mu(a,t) \le \hat \mu(a,b)$. 

Suppose $b> \hat t$, and suppose to the contrary that there exists some $b> \hat t$ such that $\mu(a,b) - \hat \mu(a,b)>0$. 
Because $\mu(a,\cdot)-\hat \mu (a, \cdot)$ is continuous, there must exist some $q \in [\hat t,b)$ such that $\mu(a,q) - \hat \mu(a,q)=0$ and crosses $0$ from above.  Consequently, 
\[
\frac{\hat u''(q)}{\hat u'(q)-\hat u'(a)}(q-\hat \mu(a,q)) < \frac{ u''(q)}{ u'(q)- u'(a)}(q- \mu(a,q)).
\]
Since $t- \mu(a,t)>0$ and $\mu(a,\cdot)$ is decreasing on $[t,1)$, we have $q- \mu(a,q)>0$. Meanwhile, $q>\hat t \ge t$ implies $\hat u''(q)<0$, $u''(q)<0$, $u'(q) >u'(b)>u'(a)$, and $\hat u'(q)>\hat u'(b) > \hat u'(a)$. Rearranging terms, we obtain:
\[
\frac{\hat u''(q)}{u''(q)} > \frac{\int_a^q \hat u''(s)\, \rmd s}{\int_a^q  u''(s)\, \rmd s} = \frac{\hat u''(a)\int_a^q \frac{\hat u''(s)}{\hat u''(a)} \, \rmd s}{ u''(a)\int_a^q \frac{ u''(s)}{ u''(a)}\, \rmd s} \ge  \frac{\hat u''(a)}{u''(a)},
\]
where the last inequality comes from the definition of likelihood ratio order and that $a<t\le \hat t$. Yet this contradicts $\hat u'' \succeq_{lr} u''$.
\end{proof}

\begin{proof}[{\bf Proof of Proposition \ref{intensityorder-S}(iii)}]

We prove the result for the change of value functions. The case for prior distributions is analogous. 
Let the optimal interval cutoffs under $u$ and $\hat u$ be $\{s_k\}_{k=1}^{N-1}$ and $\{\hat s_k\}_{k=1}^{N-1}$, respectively. 

\textit{Step 1:} We prove $\hat{s}_{N-1}\ge s_{N-1}$. Suppose to the contrary that $\hat{s}_{N-1}<s_{N-1}$. Consequently, $\hat{x}_{N}= \phi(\hat{s}_{N-1},1)<\phi(s_{N-1},1)=x_N$. It further implies that $\hat{x}_{N-1}< x_{N-1}$. To see this, suppose otherwise $\hat{x}_{N-1}\ge x_{N-1}$. We have
\begin{align*}
	s_{N-1}=\mu(x_{N-1},x_{N}) &\le \hat \mu(x_{N-1},x_{N}) \\
	&\le \hat\mu( \hat x_{N-1}, \min\{t, x_N\})\\
	&  = \hat \mu (\hat x_{N-1}+x_{N-1}-\hat x_{N-1}, \min\{t, x_N\})\\
	& \le \hat \mu(\hat x_{N-1}, \min\{t, x_N\}) + x_{N-1}-\hat x_{N-1} \\
	& \le \hat s_{N-1}  +x_{N-1} - \hat x_{N-1},
\end{align*}
where the first inequality comes from Claim \ref{propertylr-S}. The second inequality comes from the fact that $\mu(\hat x_{N-1}, \cdot)$ is increasing on $(0,t)$ and decreasing on $(t,1)$ from Claim \ref{monotonicitymu-S}. The third inequality invokes Lemma \ref{keylemma-S} and the fourth inequality again comes from the monotonicity of $\mu(\hat x_{N-1}, \cdot)$ on $(t,1)$. Hence, 
\[
\hat s_{N-1} -s_{N-1} \ge \hat x_{N-1} -x_{N-1} \ge 0,
\]
a contradiction.

Because $f$ is logconcave, we have 
\[ 
\hat{x}_{N-1}-x_{N-1} \le \max\left\{\hat{s}_{N-2}-s_{N-2}, \max \{ \hat s_{N-1} -s_{N-1},0\}\right\} \le \hat{s}_{N-2}-s_{N-2}. 
\] 
Therefore, $\hat s_{N-2} > s_{N-2}$. 

Once in the the convex region, we can apply the previous argument inductively and get
\[
\hat{x}_{k}-x_{k}\ge \hat{s}_{k}-s_{k},
\]
for $k=1,\dots,N-1$.
In particular,  
\[
\hat{x}_1-x_1\ge \hat{s}_1-s_1,
\]
which contradicts the property of logconcavity of $f$ from Lemma \ref{keylemma}. Therefore, $\hat{s}_{N-1}\ge s_{N-1}$. 

\textit{Step 2:} We prove $\hat x_{N-1} \ge x_{N-1}$. Suppose to the contrary that $\hat x_{N-1} < x_{N-1}$. By logconcavity of $f$, we must have
$s_{N-2}- \hat s_{N-2}\ge x_{N-1}-\hat x_{N-1}>0.$
On the other hand,
\begin{align*}
	s_{N-2}= \mu(x_{N-2},x_{N-1})&\le \hat \mu(x_{N-2},x_{N-1})\\
	&\le \hat \mu \left(\hat x_{N-2}+ \max \{x_{N-2}-\hat x_{N-2},0\},\hat x_{N-1}+ x_{N-1}- \hat x_{N-1}\right)\\
	&\le \hat s_{N-2}+\max\left\{\max \{x_{N-2}-\hat x_{N-2},0\},x_{N-1}-\hat x_{N-1}\right\}.
\end{align*}
Hence, $s_{N-2}-\hat s_{N-2}\le \max\{\max \{x_{N-2}-\hat x_{N-2},0\},x_{N-1}-\hat x_{N-1}\} $. By logconcavity of $\hat u''$ on $(0,\hat t)$ and Lemma \ref{keylemma}, it must be the case that 
\[
s_{N-2}-\hat s_{N-2}\le x_{N-2}-\hat x_{N-2}.
\]
Similar to Step 1, we can conduct this argument inductively and obtain $s_1-\hat s_1 \le x_1 - \hat x_1$,
which again contradicts the logconcavity of $f$. Therefore, we must have $\hat{x}_{N-1}<x_{N-1}$. 

\textit{Step 3:} Repeat the similar argument inductively, we prove that $\hat x_k \ge x_k$ for all $k=1,\ldots,N$ and $\hat s_k \ge s_k$ for all $k=1,\ldots,N-1$. 
\end{proof}

\begin{proof}[{\bf Proof of Proposition \ref{intensityorder-S}(iv)}]
Again, we prove the result for the change of value functions. The case for change of prior distributions is similar. 

Define the sequence,
\begin{equation*}
	\{\delta_k\}_{k=1}^{2N-1} :=\{x_1-\hat x_1, s_1-\hat s_1, x_2-\hat x_2, \ldots,  s_{N-1}-\hat s_{N-1}, x_N-\hat x_N \}.
\end{equation*}
Note that $\delta_{2N-2} \ge 0$ implies $\delta_{2N-1} \ge 0$, since $s_{N-1} \ge \hat{s}_{N-1}$ implies $x_N \ge \hat x_N$.

Also note that $u'' \succeq_{uv} \hat u''$ implies either $t \le \hat t \le p $ or $p \le \hat t\le t$. Suppose $t \le \hat t \le p $, then $x_{N-1}<s_{N-1} \le t\le p$ and $\hat x_{N-1} <\hat s_{N-1} \le \hat t\le p$. 
Let $j$ be the first non-negative term in the sequence $\{ \delta_k\}_{k=1}^{2N -1}$. For all $j=1, \dots, 2N-3$, by the same argument in Proposition \ref{csforcutoffsuv}, $\delta_j \ge 0$ implies $\delta_{k} \ge 0 $ for all $k=j+1,\dots, 2N-2$. 
Hence $\{\delta_k\}_{k=1}^{2N-1}$ is single-crossing from below.

From now on, we focus on the case where $ p \le \hat t \le t$. There are five cases to consider.

\textit{Case 1:} $ x_{N-1} \le \hat x_{N-1} \le p $. We want to show $\delta_{k}\le 0$ for all $k\le 2N-3$. Since $\hat{x}_{N-1} \le p \le \hat{t}$ and $x_{N-1} \le p \le t$, if there exists some $j\le 2N-4$ such that $\delta_j> 0$, then by the argument from Proposition \ref{csforcutoffsuv}, $\{ \delta_{k}\}_{k=j}^{2N-3} $ is increasing and therefore all terms in the sequence are strictly positive, which contradicts the premise, $x_{N-1}-\hat{x}_{N-1}=\delta_{2N-3} \le 0$. 

\textit{Case 2:}  $\hat x_{N-1} < x_{N-1} \le p$. We want to show $s_{N-1}-\hat s_{N-1} \ge 0$. Suppose to the contrary that 
$s_{N-1}< \hat s_{N-1} $. Then it must be the case that $s_{N-2} \ge \hat{s}_{N-2}$; for otherwise $x_{N-1}= \phi (s_{N-2}, s_{N-1})$ must be strictly lower than $\hat x_{N-1}= \phi (\hat s_{N-2}, \hat s_{N-1})$. Hence, we have
\begin{align*}
	x_{N-1} = \phi (s_{N-2}, s_{N-1}) &=   \phi ( \hat{s}_{N-2} +  s_{N-2}-\hat s_{N-2},   s_{N-1})\\
	& < \phi ( \hat{s}_{N-2} +  s_{N-2}-\hat{s}_{N-2},   \hat s_{N-1})\\
	& \le \hat{x}_{N-1} + s_{N-2}- \hat{s}_{N-2},
\end{align*}
which implies $s_{N-2}- \hat s_{N-2} > x_{N-1}- \hat x_{N-1}$. On the other hand, the argument in Proposition \ref{csforcutoffsuv} and $\hat x_{N-1} < x_{N-1} \le p$ indicates that $x_{N-1}- \hat x_{N-1} \ge  s_{N-1} -\hat s_{N-1} \ge 0$. This is a contradiction.

\textit{Case 3:} $x_{N-1} \le p \le \hat{x}_{N-1}$ with $x_{N-1}< \hat{x}_{N-1}$. We want to show $\delta_k  \le 0$ for all $k \le 2N-4$. It suffices to prove $\delta_{2N-4} = s_{N-2} - \hat s_{N-2} < 0$, then we can invoke the argument in Proposition \ref{csforcutoffsuv} to show $\delta_k \ge 0$ for all $k\le 2N-4 $. Observe that  
\begin{align*}
	s_{N-2} = \mu (x_{N-2}, x_{N-1})    
	& \le \mu (x_{N-2},p) \\
	& \le \hat \mu (x_{N-2},p) \\
	& \le \hat \mu\left(\hat x_{N-2}+ \max\{ x_{N-2}-\hat x_{N-2},0\}, \hat x_{N-1}\right)  \\
	& \le \hat s_{N-2} + \max \{ 0, x_{N-2}- \hat x_{N-2}\} .
\end{align*}
Therefore, we have $s_{N-2} -\hat s_{N-2} <\max \{ 0, x_{N-2}- \hat x_{N-2}\} $ (the 
inequality is strict because either $\hat{x}_{N-1} >p$ or ${x}_{N-1} <p$). 
If $s_{N-2} -\hat{s}_{N-2} \ge 0$, then we must have $s_{N-2} -\hat s_{N-2} < x_{N-2} - \hat x_{N-2}$, which contradicts the fact that $\{\delta_k\}_{k=j}^{2N-3}$ is nondecreasing for any $j$ such that $\delta_j\ge0$, as shown in Proposition \ref{csforcutoffsuv}.

\textit{Case 4:} $\hat x_{N-1} \le p \le {x}_{N-1}$ with $\hat x_{N-1} < x_{N-1}$. First, for the subsequence $\{ \delta_k\}_{k=1}^{2N-5}$, $\delta_k \ge 0 $ implies $\delta_{k+1} \ge 0$, by the same reasoning in Proposition \ref{csforcutoffsuv}. The case for $i=2N-4$ is automatically satisfied by the premise that $\hat x_{N-1} < x_{N-1}$. It remains to be shown that $s_{N-1} -\hat s_{N-1} \ge 0 $.  Suppose on the contrary that $s_{N-1} < \hat s_{N-1} $ (and thereby $x_N < \hat x_{N}$).

If $t \le x_N < \hat x_{N-1}$, we have 
\[
\hat s_{N-1}= \hat \mu(\hat x_{N-1}, \hat x_{N}) 
\le \hat \mu(p, \hat x_{N}) 
\le \mu (p, \hat x_{N})  
< \mu (x_{N-1}, \hat x_{N})
< \mu(x_{N-1},x_N)=s_{N-1},
\]
a contradiction. 
If $x_N < t$,
\begin{align*}
	\hat s_{N-1} = \hat \mu(\hat x_{N-1}, \hat x_{N})
	&\le \hat \mu(p, \hat x_{N}) \\
	& \le \mu (p, \hat x_{N})  \\
	& < \mu (x_{N-1} , \min\{\hat x_N,t \})\\
	& \le \mu (x_{N-1} , x_N + \min\{\hat x_N-x_N,t-x_N \})\\
	&\le s_{N-1}+\min\{\hat x_N-x_N,t-x_N \}.
\end{align*}
Therefore, $\hat s_{N-1}-s_{N-1}< \min\{\hat x_N-x_N,t-x_N \} \le \hat x_N-x_N$.
On the other hand, by logconcavity of $f$, $\hat s_{N-1} - s_{N-1} \ge \hat x_N -x_N$, a contradiction.

\textit{Case 5:}  $p \le \min \{ \hat x_{ N-1}, x_{ N-1} \}$. For the subsequence $\{ \delta_k\}_{k=1}^{2N-4}$, $\delta_k \ge 0$ implies $\delta_{k+1}\ge 0 $, by the same argument in Proposition \ref{csforcutoffsuv}.  It remains to be shown that $s_{N-1} < \hat s_{N-1}$ (and thereby $x_{N} < \hat x_{N}$) implies $x_{N-1} < \hat x_{N-1}$. Observe that  $p \le \min \{ \hat x_{ N-1}, x_{ N-1} \}$ implies 
$\hat s_{N-1}   = \hat \mu(\hat x_{N-1}, \hat x_N)   \le \mu(\hat x_{N-1}, \hat x_N) $.

If $\hat x_N \ge t$, then
\begin{align*}
	\hat s_{N-1}  & \le \mu(\hat x_{N-1}, \hat x_N)  \\
	& < \mu(\hat x_{N-1},  \max \{ x_N,t\})  \\
	&= \mu\left(x_{N-1} + \max\{ \hat x_{N-1} -x_{N-1},0\},  x_{N} +\max \{ 0,t-x_N\}\right) \\
	& \le  s_{N-1}+ \max\left\{\max\{0, \hat x_{N-1} -x_{N-1} \} ,  \max \{ 0,t-x_N\}\right\}\\
	& = s_{N-1}+\max\{0, \hat x_{N-1} -x_{N-1} \}.
\end{align*}
Hence, it must be the case that $\hat x_{N-1} -x_{N-1} > 0 $ as desired.
If $\hat{x}_N < t$, then 
\begin{align*}
	\hat s_{N-1} & \le \mu(\hat x_{N-1}, \hat x_N)  \\
	& < \mu( x_{N-1} + \max\{\hat x_{N-1} - x_{N-1},0\},  x_N +\hat x_N-  x_N) \\
	& \le s_{N-1}+ \max \{ \hat x_{N-1} -x_{N-1} ,\hat x_N-  x_N\}.
\end{align*}
By logconcavity, we have $\hat x_N- x_N \le \hat s_{N-1}- s_{N-1}$. So it must be the case that $\hat x_{N-1} - x_{N-1} \geq  \hat s_{N-1} - s_{N-1}>0$.
\end{proof}

\end{document}